\documentclass[10pt]{article}
\usepackage{dcolumn}% Align table columns on decimal point
\usepackage{bm}% bold math
\usepackage{hyperref}
\usepackage{verbatim}       % Defines \begin{comment}, \end{comment}

\usepackage{color}
\usepackage{graphicx}

\usepackage{amsthm}
\usepackage{amssymb}

\usepackage{bm}
\usepackage{amstext}
\usepackage{psfrag}
\usepackage{amsmath}
\usepackage{amsfonts}
\usepackage{units}
\usepackage{mathrsfs}

\newcommand{\dd}{{\rm d}}

\newcommand{\bd}{\begin{definition}}                %inizia definizione
\newcommand{\ed}{\end{definition}}                  %fine definizione
\newcommand{\bc}{\begin{corollary}}                 %inizia corollario
\newcommand{\ec}{\end{corollary}}                   %fine corollario
\newcommand{\bl}{\begin{lemma}}                     %inizia lemma
\newcommand{\el}{\end{lemma}}                       %fine lemma
\newcommand{\bp}{\begin{proposition}}            %inizia proposizione
\newcommand{\ep}{\end{proposition}}                %fine proposizione
\newcommand{\bere}{\begin{remark}}                  %inizia osservazione
\newcommand{\ere}{\end{remark}}                     %fine oservazione

\newcommand{\bt}{\begin{theorem}}
\newcommand{\et}{\end{theorem}}

\newcommand{\be}{\begin{equation}}
\newcommand{\ee}{\end{equation}}

\newcommand{\bit}{\begin{itemize}}
\newcommand{\eit}{\end{itemize}}
\newtheorem{theorem}{Theorem}[section]
\newtheorem{corollary}[theorem]{Corollary}
\newtheorem{lemma}[theorem]{Lemma}
\newtheorem{proposition}[theorem]{Proposition}
\theoremstyle{definition}
\newtheorem{definition}[theorem]{Definition}
\theoremstyle{remark}
\newtheorem{remark}[theorem]{Remark}
\newtheorem{example}[theorem]{Example}

\begin{document}

\title{The representation of  spacetime through time functions}

\author{E. Minguzzi\footnote{Dipartimento di Matematica, Universit\`a degli Studi di Pisa,  Largo
B. Pontecorvo 5,  I-56127 Pisa, Italy. E-mail:
ettore.minguzzi@unipi.it, ORCID:0000-0002-8293-3802}}

%\date{April 2024, repru.tex}
\date{}
\maketitle

\begin{abstract}
The properties of the stable distance over stable spacetimes are used as a reference to propose a simplified, abstract notion of spacetime.
Our analysis establishes that the fundamental structures of spacetime, namely its topology, causal order, and (upper semi-continuous) Lorentzian distance, can be derived from a general and minimalistic set of axioms.
Specifically, it is shown that spacetime can be represented as nothing more than a family of functions defined over an arbitrary set, the functions being a posteriori interpreted as rushing time functions.  The proof makes use of the product trick which reduces causality and metricity to causality in a space with one additional dimension, so leading to a  unification for the notions of time function and proper time. Ultimately, our results show that time fully characterizes spacetime.
\end{abstract}

\setcounter{secnumdepth}{2}
\setcounter{tocdepth}{2}
\tableofcontents

\section{Introduction}

In this work we shall explore a functional approach to spacetime which might facilitate the unification of general relativity with quantum mechanics, and hence of gravity with the other fundamental interactions.

Our study is motivated by the need to understand time and its role in physics. As is well known, time is treated in fundamentally different ways in quantum mechanics and general relativity. In quantum theory, it appears as an external, absolute parameter; even when Lorentz invariance is implemented, time remains a rigid construct, as each inertial observer is associated with a distinct time coordinate, with the family of inertial observers related by Lorentz transformations. This contrasts sharply with general relativity, where the term ``time" encompasses two distinct notions: {\em proper time}, which is fundamentally a metric concept, and {\em time functions}, which are causal constructs.

The persistent difficulty in reconciling these treatments of time within unification frameworks has led many scholars to adopt a skeptical attitude toward time itself. Some regard it as an emergent, illusory concept arising from more fundamental structures. These timeless approaches to physics, often inspired by time-independent treatments in classical mechanics (e.g., the Jacobi metric), recur throughout the physics literature. While intriguing, they have thus far yielded limited progress.

Conversely, we have explored over the years an  opposing view \cite{minguzzi09c,minguzzi13e,minguzzi17d}: that time, far from being emergent,  is actually the main  building block, i.e., the foundational fabric, of the arena of physical theories — the {\em spacetime} itself.
 A key step lies in recognizing that the metric and causal aspects of time can be unified through the ``product trick" \cite{minguzzi17,minguzzi17d,minguzzi18b}, a technical device with profound physical significance. This approach serves as the central idea and tool enabling the results presented in this work.

In previous work we had already established that the functional approach to spacetime is feasible  in the context of closed cone structures, and in particular, stable Lorentz-Finsler spaces \cite[Thm.\ 4.6]{minguzzi17}. In this setting the spacetime is described by a manifold, while the light cones and the Lorentz(-Finsler) metric have upper semi-continuous regularity. We showed that on stable spacetimes the (Seifert) causal order, the manifold topology and the (stable) Lorentzian distance can be recovered from the set of $F$-steep temporal functions, that is, time functions that increase sufficiently rapidly over causal curves. This study, which  will be recalled in Section \ref{cmotp}, will mostly serve for motivation and as a reference.

It is natural to ask if a similar functional characterization of spacetime can be obtained in an abstract setting in which the spacetime is no more regarded as a manifold. We shall indeed give a positive answer to this question provided by ``spacetime'' we understand a specific mathematical object. Indeed, part of the problem will be that of identifying the type of spacetime we wish to recover.

We shall observe that spacetime is a primitive and quite abstract notion,
and that it can be induced by any family of point-distinguishing functions over a set (see Sec.\ \ref{vbqp}).

Since the functions in the family are subsequently recognized to be rushing (the metric analog of $F$-steepness), it becomes natural to ask if the spacetime can be recovered from the family of rushing functions over it. We are able to prove this type of converse under  (suitable) topological assumptions  (Thm.\ \ref{mai}). Thus, ultimately, in most cases of interest the spacetime can be identified with a distinguishing family of functions over a set.

Paraphrasing the words of Ray Cummings and John Wheeler, time not only keeps everything from happening at once, it is what makes events discernible.

This work aligns with the recent developments in low regularity Lorentzian geometry \cite{chrusciel12,kunzinger13b,minguzzi13d,graf18,sbierski18,mccann18,mondino18,heveling21b,burtscher22,burgos23}  and spacetime geometry \cite{kunzinger18,muller19,ake20,muller22b,minguzzi22,sakovich24,cavalletti20,beran23,braun23b,minguzzi24b}.
However, its origins trace back to an earlier research direction  which  started with work by the author \cite{minguzzi17,canarutto19}  on the Lorentzian distance formula.  That earlier investigation was itself motivated by studies related to Connes’ program for the unification of fundamental forces, particularly efforts toward a Lorentzian generalization \cite{besnard09,franco13,besnard15,franco18}.

The methods we shall use belong to Nachbin's theory of topological preordered spaces \cite{fletcher82,nachbin65}, subsequently extended in some work by the author \cite{minguzzi11f,minguzzi12d}, to deal with the non-compact case. For some recent papers on this order-topological framework, see e.g.\ \cite{yamazaki23,yamazaki24,ok25,bosi25}

In this respect, we recall that a {\em topological preordered space}  is a triple $(X,\mathscr{T},\le)$ where $X$ is a set, $\mathscr{T}$ is a topology on $X$, and $\le$ is a {\em preorder}, i.e.\ a reflexive and transitive relation on $X$. It is called {\em order} if it is {\em antisymmetric}: $x\le y$ and $y\le x$ $\Rightarrow x=y$. We speak of {\em closed preordered space} if the graph $G(\le)\subset X\times X$ of $\le$ is closed in the product topology $\mathscr{T}\times \mathscr{T}$. With $i(x):=\{y: x\le y\}$ we denote the {\em increasing hull} (or {\em future}), while with $d(y):=\{x: x\le y\}$ we denote the {\em decreasing hull} (or {\em past}). Similarly, we define for $S\subset X$, $i(S):=\cup_{x\in S} i(x)$ and  $d(S):=\cup_{x\in S} d(x)$. A subset $S$ is {\em convex} if $S=i(S)\cap d(S)$. It is called {\em increasing} if $i(S)=S$ and {\em decreasing} if $d(S)=S$. The complement of a decreasing (increasing) set is increasing (resp.\ decreasing). In a closed preordered space, if $K$ is compact then $i(K)$ and $d(K)$ are closed \cite[Prop.\ 4]{nachbin65}.

A closed preordered space is {\em normally preordered} if for every two disjoint closed sets $A$ and $B$, respectively  decreasing and increasing, we can find an open decreasing set $U$ and an open increasing set $V$ such that $A\subset U$, $B\subset V$, $U\cap V=\emptyset$.  A function $f$ is said to be {\em isotone} if $x\le y \Rightarrow f(x) \le f(y)$.

Nachbin's generalization of Urysohn's lemma states that in a normally preordered space, for $A$ and $B$ as above, there is a continuous isotone function $f:X\to [0,1]$ such that $f^{-1}(0)\supset A$, $f^{-1}(1)\supset B$.

A closed preordered space $(X, \mathscr{T}, \le)$ is {\em completely regularly preordered} if the topology and preorder can be recovered from the set of continuous isotone functions. In other words, \( \mathscr{T} \) is the initial topology of the family of continuous isotone functions, and we have $x \le y$ if and only if $f(x) \le f(y)$ for every continuous isotone function $f$.

The topological preordered space is said to be {\em convex} if the open increasing  and the open decreasing sets form a subbasis for the topology. The completely regularly preordered spaces are convex.\footnote{Indeed, if $O\ni x$ is open we can find continuous isotone functions $f_i$, $i=1,\ldots, k$ and $g_j$, $j=1,\ldots, l$, such that $x\in \cap_j g_j^{-1}((-\infty,0))\cap_i  f_i^{-1}((0,+\infty))\subset O$. Now set $U=\cap_j   g_j^{-1}((-\infty,0))$ and $V=\cap_i f_i^{-1}((0,+\infty))$, then $U$ is open decreasing, $V$ is open increasing and $x\in U\cap V \subset O$.} The normally preordered spaces which are convex are completely regularly preordered.

It is important to note that convexity is a global property and should not be confused with the less stringent condition of {\em weak convexity}, which requires the topology to have a basis of open convex sets, or with the even more relaxed condition of {\em local convexity}, where every point has a neighborhood system consisting of convex sets. Although convexity entails weak convexity, and weak convexity entails local convexity, the reverse implications are not always valid. This is due to the fact that not every convex neighborhood is open, and not every open convex set can be expressed as the intersection of an open increasing set and an open decreasing set.

A {\em quasi-uniformity} is a filter $\mathscr{U}$ on the diagonal of $X \times X$ that satisfies the condition: for every $U \in \mathscr{U}$, there exists $V \in \mathscr{U}$ such that $V \circ V \subset U$. Unlike a uniformity, the symmetry condition ``if $U \in \mathscr{U}$ then $U^{-1} \in \mathscr{U}$'' where $U^{-1}=\{(x,y): (y,x)\in U\}$, is not imposed. As a result, $G = \bigcap \mathscr{U}$ is a preorder on $M$,  while the symmetrized uniformity $\mathscr{U}^* = \{ U \cap U^{-1} : U \in \mathscr{U} \}$ induces a topology $\mathscr{T}(\mathscr{U}^*)$. In essence, every quasi-uniformity induces a topological preordered space $(X,\mathscr{T}(\mathscr{U}^*), \bigcap \mathscr{U})$, where the topology is Hausdorff if and only if the preorder is an order. The quasi-uniformity is a notion that unifies topology and preorder showing that they are really two aspects of the same entity.

The completely regularly preordered spaces coincide with the quasi-uni\-for\-mi\-za\-ble spaces \cite[Thm.\ 9, p.\ 67]{nachbin65}.

In this work we shall make repeated use of the product trick, introduced by the author in \cite{minguzzi17}, and further popularized in \cite{minguzzi17d,minguzzi18b}, to unify causality and metricity of a space as causality in a spacetime with one additional dimension. In one sense, the proofs we provide are simpler than those of \cite{minguzzi17} because we do not have the additional complication of having to smooth the representing functions. On the other hand, some proofs are more complicated because we do not work  with manifolds with their good topological properties. Moreover, the causal order we deal with is not necessarily compactly generated, i.e.\ generated from local information, as it happens for causality in cone structures, a fact that prevents a complete analogy with the results in \cite{minguzzi17}.

Although we make often reference to  \cite{minguzzi17} where some  key ideas were first introduced, the paper is essentially self-contained. The results of this work have been announced in \cite{minguzzi25}.

\section{The stable Lorentzian distance (manifold case)} \label{cmotp}

This section can be skipped on first reading. It is meant to provide context, motivation, and reference for an analogy used later in  this work.

In a traditional smooth Lorentzian manifold $(M,g)$ every (absolutely continuous) causal curve $\sigma: I \to M$, is assigned a length (proper time) $\ell(\sigma)=\int_I \sqrt{-g(\dot \sigma, \dot \sigma)} \dd t$.
Two events $x,y\in M$ are said to be causally related, written $(x,y)\in J$,  $y\in J^+(x)$ or $x\le y$,  if $x=y$ or there is a future-directed causal curve connecting $x$ to $y$.

The {\em Lorentzian distance} between two events $x$ and $y$, $x\le y$, is then defined as the supremum of the lengths of the connecting causal curves
\[
d(x,y):=\sup_\sigma \ell (\sigma),
\]
where $\sigma: [0,1]\to M$, $\sigma(0)=x$, $\sigma(1)=y$. We set $d(x,y)=0$ if $(x,y)\notin J$. This two-point function $d$ satisfies the reverse (or Lorentzian) triangle inequality: for $x\le y \le z$
\[
d(x,z)\ge d(x,y)+d(y,z).
\]
We see that $d$ is somewhat analogous to the notion of {\em metric/distance} known from topology, so we might say that the Lorentzian distance (which physically encodes proper time) is a {\em metric notion}. In the standard smooth Lorentzian framework function $d$ is lower semi-continuous, while it is continuous only under further conditions, such as global hyperbolicity.

We refer the reader to \cite{hawking73,beem96,minguzzi18b} for the traditional description of spacetime via time-oriented Lorentzian manifolds.

Let us take a first step into more abstract settings, so departing from general relativity.
To address the metric aspects of spacetime, particularly proper time, it is natural, in any abstract approach to spacetime, to introduce a Lorentzian distance function \( d \).
%Other properties of spacetime including the proper length of curves, ideally, would be derived from $d$.

 A key question arises: what properties should \( d \) satisfy in an abstract setting? To explore this, we first examine the behavior of \( d \) in low-regularity settings with weak causality assumptions, while still assuming the spacetime is a manifold. A useful reference is the theory of closed/proper Lorentz-Finsler spaces developed in \cite{minguzzi17}. In this framework, the causal cones are defined by an upper semi-continuous distribution \( x \mapsto C_x \subset T_xM \setminus 0 \) of closed, sharp, non-empty convex cones. These form a convex, sharp, closed subbundle of the slit tangent bundle, defining a so-called closed cone structure. The metric aspects are governed by a Finsler fundamental function \( F: C_x \to [0, \infty) \), which satisfies specific properties introduced via the {\em product trick}. This trick involves studying the spacetime \( M^\times = M \times \mathbb{R} \), defining a translationally invariant cone structure using \( F \), and imposing upper semi-continuity on this cone distribution.

At each point \( P = (p, z) \) of \( M^\times \), the cone \( C_P^\downarrow \) is defined as a subset of \( T_pM \setminus 0 \times \mathbb{R} \) given by the hypograph of the function (with an alternative symmetrized version omitted here) \cite[Eq.\ (3.16)]{minguzzi17}:
\begin{equation} \label{curx}
F^\downarrow =
\begin{cases}
F(p), & \text{if } v \in C_p, \\
-\infty, & \text{if } v \notin C_p.
\end{cases}
\end{equation}
The triple \( (M,C,F) \) is called a {\em closed Lorentz-Finsler space} if \( (M^\times, C^\downarrow) \) forms a closed cone structure. A {\em proper Lorentz-Finsler space} adds an extra condition to prevent the cones and \( F \) from collapsing in certain directions. Within this more specialized structure, the chronological relation can be introduced, though it is generally absent in broader settings.

For a closed Lorentz-Finsler space (or its closed cone structure), stable causality can still be defined as the stability of causality under sufficiently narrow cone enlargements. The Seifert relation \( J_S:=\cap_{C'>C} J_{C'} \) can also be defined as the intersection of all the causal relations obtained enlarging the cones, with its antisymmetry being equivalent to stable causality and the existence of time functions. In this case, $J_S$ is the  smallest closed, reflexive and transitive relation containing $J$, i.e.\  \( K = J_S \) \cite[Thm.\ 3.16]{minguzzi17}. Global hyperbolicity can similarly be defined, though the compactness of causal diamonds must be replaced by the preservation of compactness under the causally convex hull $S\mapsto c(S):=i(S)\cap d(S)$. This allows the recovery of a hierarchy of causality properties, suggesting that the chronological relation may be less central than previously thought.

On the metric side, globally hyperbolic spacetimes exhibit a finite Lorentzian distance \cite[Prop.\ 2.26]{minguzzi17} that is upper semi-continuous \cite[Thm.\ 2.60(g)(d)]{minguzzi17}. Full continuity of \( d \) requires more specialized structures, such as proper Lorentz-Finsler spaces and stronger regularity conditions (e.g., local Lipschitzness) \cite[Thm.\ 2.53]{minguzzi17}, see however \cite{ling24} for an approach which mitigates these difficulties in the $C^0$ setting. In summary, under weak causality/low regularity conditions, upper semi-continuity and finiteness may be lost, and lower semi-continuity demands additional assumptions.

Does low regularity and weak causality render all continuity properties for a Lorentzian distance  invalid? Not entirely. A natural Lorentzian distance, denoted \( D \), emerges in this regime. Unlike \( d \), function \( D \) is upper semi-continuous and does not require additional assumptions like global hyperbolicity or local Lipschitzness. Moreover, \( D \) satisfies the reverse triangle inequality for a broader relation than the causal relation \( J \), namely \( J_S \). Historically, \( d \) has received more attention, but if \( D \) had been discovered earlier, the focus might have shifted. The value \( D(p, q) \) is defined by slightly enlarging the cones and considering the causal curves (for the enlarged cones) that connect \( p \) and \( q \). By also enlarging the {\em indicatrices}\footnote{The non-compact unit ball is enlarged to  contain the original one.}  within the cones, proper time is computed for each enlargement, and the infimum of these values defines \( D(p, q) \), termed the {\em stable Lorentzian distance}. This function is upper semi-continuous, satisfies the reverse triangle inequality for \( J_S \), and under stable causality, \( D(p, p) = 0 \). Additionally, \( d \leq D \), and under global hyperbolicity, \( d = D \).

An important concept associated with \( D \) is that of a {\em stable spacetime}, which lies between global hyperbolicity and stable causality \cite[Thm.\ 2.63]{minguzzi17}. A stable spacetime is defined by the finiteness of \( D \). A theorem \cite[Thm.\ 2.61]{minguzzi17} ensures that stable spacetimes remain stable under cone and indicatrix enlargements, preserving both causality and the finiteness of \( D \).

The significance of stable spacetimes is highlighted by two key results. First, under sufficient differentiability (\( C^3 \)), stable  spacetimes $(M,g)$ can be characterized as Lorentzian submanifolds (possibly with boundary) of Minkowski spacetime \( \mathbb{L}^n \) for some dimension \( n \) \cite[Thm.\ 4.13]{minguzzi17}. Second, the geometry of the spacetime can be represented through suitable sets of functions. For a closed cone structure under stable causality, the Seifert relation \( J_S \) can be recovered using continuous isotone functions, as established by Auslander-Levin's theorem and related results \cite{auslander64,levin83,minguzzi11c}. The topology can also be recovered as the initial topology of these functions, a property known as complete order regularity \cite{minguzzi12d}.

A deeper goal is to recover the triple topology, order, and Lorentzian metric from a family of functions. Parfionov and Zapatrin \cite{parfionov00} conjectured the {\em Lorentzian distance formula}:
\[
d(p, q) = \inf\{ [f(q) - f(p)]^+ : f \in \mathscr{S} \},
\]
where \( \mathscr{S} \) is the family of \( F \)-steep temporal functions and \( a^+ := \max\{0, a\} \). We recall that a function \( f: M\to \mathbb{R} \) is {\em temporal} if it is \( C^1 \) and \( \dd f \) is positive on the future causal cone \( C_p \) for every \( p \), and {\em \( F \)-steep} if \( \dd f(v) \geq F(v) \) for every \( v \in C \). This formula was proven for globally hyperbolic spacetimes \cite[Thm.\ 4.67]{minguzzi17} \cite{minguzzi17d}, but it is more naturally suited to weaker causality conditions. The infimum in the formula implies upper semi-continuity, aligning with the properties of \( D \). Indeed, \( D \) can be recovered in general, and under global hyperbolicity, \( D = d \) \cite[Thm.\ 2.60(g)]{minguzzi17}. We proved the following theorem \cite[Thm.\ 4.6]{minguzzi17}:

\begin{theorem} \label{aas}
Let \( (M, F) \) be a closed Lorentz-Finsler space, and let \( \mathscr{S} \) be the family of smooth \( F \)-steep temporal functions. The Lorentz-Finsler space \( (M, F) \) is stable if and only if \( \mathscr{S} \) is non-empty. In this case, \( \mathscr{S} \) represents:
\begin{itemize}
\item[(i)] the manifold topology, as it is the initial topology of the functions in \( \mathscr{S} \);
\item[(ii)] the order \( J_S \), i.e., \( (p, q) \in J_S \Leftrightarrow f(p) \leq f(q) \) for all \( f \in \mathscr{S} \);
\item[(iii)] the stable distance, via the distance formula:
\begin{equation}
D(p, q) = \inf \big\{ [f(q) - f(p)]^+ : f \in \mathscr{S} \big\}.
\end{equation}
\end{itemize}
\end{theorem}
This result will be generalized to the non-manifold case in Sec.\ \ref{cnxp} (see Thm.\ \ref{mai}).

In the following section we propose a definition of spacetime that abstracts the properties of the stable distance in closed Lorentz-Finsler spaces \cite{minguzzi17} (see Theorem 2.6 of \cite{minguzzi17}, properties (a), (c), (d); under stable causality (h) is added besides the antisymmetry of the relation;   for stability see  \cite[Def.\ 2.29]{minguzzi17}). The  correspondence with the notation of that paper is $X \leftrightarrow M$, $d\leftrightarrow D$, $\le \leftrightarrow J_S$.

\section{A definition of spacetime} \label{vbqp}

For us the spacetime will be the following object
\begin{definition}
A {\em spacetime} $M$ is a quadruple $(X,\mathscr{T},\le, d)$, where  $(X,\mathscr{T},\le)$ is a  closed preordered space, and $d:X\times X\to [0,\infty]$ is an $\mathscr{T}\times \mathscr{T}$-upper semi-continuous function such that, $x\nleq y$ implies $d(x,y)=0$, and for every triple $x\le y\le z$, we have (reverse triangle inequality)
\[
d(x,y)+d(y,z)\le d(x,z).
\]
We say that it is {\em weakly stably causal}\footnote{One can be tempted to call this property just {\em causality} as there are no other causal relations, still, I prefer to retain, at least in this work, the terminology that best clarifies the transition to the manifold case. The adjective ``weak" will be removed when we shall add the property of local convexity.} if $\le$ is an order and $d(x,x)=0$ for every $x\in X$. If, additionally, $d$ is finite, we say that it is {\em weakly stable}. We also call $(X,\mathscr{T},\le)$ the {\em causal structure} of the spacetime. It is {\em weakly stably causal} if $\le$ is an order.
\end{definition}
Note that for every $x\in X$, $d(x,x)=0$ or $d(x,x)=+\infty$ (by the reverse triangle inequality). Observe that the definition of weak stable causality for the causal structure is obtained from that of the spacetime dropping the condition on $d$, as it does not enter the causal structure. A causal structure $(X,\mathscr{T},\le)$ is weakly stably causal iff the spacetime $(X,\mathscr{T},\le,d)$ obtained through the trivial choice $d=0$, is weakly stably causal.

A function $d$ that satisfies the above properties is called a (stable) {\em Lorentzian distance}. The letter $d$ is also used for the decreasing hull of a preorder, but hopefully the adopted notation shall not cause confusion.

Given a spacetime $(X,\mathscr{T},\le, d)$, the closed ordered space $(X,\mathscr{T},\le)$ will be called the {\em causal structure} of the spacetime. As a notation, for simplicity, we might denote the graph of $\le$ with $\le$ itself or $J$ (though this concept is the low regularity counterpart of $K$).

\begin{definition} \label{cptr}
A {\em spacetime}  is a triple $(X,\mathscr{T},\tau)$, where   $\tau:X\times X\to \{-\infty\}\cup [0,\infty]$ is an $\mathscr{T}\times \mathscr{T}$-upper semi-continuous function such that for every $x\in X$, $\tau(x,x)\ge 0$, and  for every $x,y,z\in X$
\[
\tau(x,y)+\tau(y,z)\le \tau(x,z),
\]
with the convention $-\infty+\infty=-\infty$.
We say that it is {\em weakly stably causal} if $\textrm{min}\{\tau(x,y),\tau(y,x)\} \ge 0 \Rightarrow x=y$ and for every $x\in X$, $\tau(x,x)=0$. If, additionally, $\tau<+\infty$ , we say that it is {\em weakly stable}.
\end{definition}

The two notions of spacetime are equivalent. To pass from the former to the latter set
\begin{equation} \label{cury}
\tau(x,y) =
\begin{cases}
d(x,y), & \text{if } x\le y, \\
-\infty, & \text{if } x\nleq y.
\end{cases}
\end{equation}
To pass from the latter to the former set $\le=\{\tau\ge 0\}$, $d(x,y)=\textrm{max}\{0,\tau(x,y)\}$.

The function $\tau$ might be called {\em time separation} to distinguish it from $d$. It is really the metric analog of $F^\downarrow$, see Eq.\ (\ref{curx}), where $d$ is the metric analog of $F$. For this reason, $\tau$ could also be denoted $d^\downarrow$.

\begin{proof}[Proof of the equivalence]
Suppose that $(X,\mathscr{T},\le, d)$ is a spacetime. Let us prove that $\tau$ is upper semi-continuous at $(x,y)\notin\, \le$.  Indeed,  $(x,y)$ belongs to the open set $X\times X\backslash \le$ where $\tau$ is constant hence continuous. Suppose $(x,y)\in \,\le$, if $d(x,y)=+\infty$, as $\tau(x,y)=+\infty$ upper semi-continuity of $\tau$  is clear, so let us assume $d(x,y)<+\infty$, then by the upper semi-continuity of $d$ for every $\epsilon>0$ we can find an open set $O\ni (x,y)$ such that $d(x',y')\le d(x,y)+\epsilon$, for every $(x',y')\in O$ which implies $\tau(x',y')\le d(x',y')\le d(x,y)+\epsilon= \tau(x,y)+\epsilon$.

Since $x\le x$, $\tau(x,x)=d(x,x)\ge 0$.

The validity of the reverse triangle inequality for $\tau$  when some pair does not belong to $\le$ is clear except possibly  when  the pair is $(x,z)$. In this case $(x,y)$ or $(y,z)$ do not belong to $\le$ otherwise $(x,z)$ would belong to it, and so the left hand side is $-\infty$ and the inequality holds.
If all pairs belong to $\le$ the inequality follows from the reverse triangle inequality for $d$.

For the converse, suppose that  $(X,\mathscr{T},\tau)$ is a spacetime. As $\tau$ is upper semi-continuous the set $\le:=\{\tau \ge 0\}$ is closed. The reverse triangle inequality for $\tau$ implies that $\le$ is transitive. The condition $\tau(x,x)\ge 0$ implies  $x\le x$, i.e.\ reflexivity of $\le$ (by the reverse triangle inequality $\tau(x,x)=0$ or $\tau(x,x)=+\infty$).

Suppose $x\le y\le z$, then the reverse triangle inequality for $\tau$ reads $d(x,y)+d(y,z)\le d(x,z)$. Finally, $d=\textrm{max}\{0,\tau\}$ is upper semi-continuous as the maximum of two upper semi-continuous functions (upper semi-continuity is equivalent to closure of the hypograph and the maximum implies the new hypograph is the union of the two starting closed hypographs hence closed).

Let us come to the equivalence of the properties defining weak stable causality.
Suppose  $(X,\mathscr{T},\le, d)$ is weakly stably causal.
If $\textrm{min}\{\tau(x,y),\tau(y,x)\} \ge 0$ then both terms are non-negative,  hence by the definition of $\tau$, $x\le y$ and $y\le x$ which implies $x=y$ by the antisymmetry of $\le$. We have shown the validity of the implication  $\textrm{min}\{\tau(x,y),\tau(y,x)\} \ge 0 \Rightarrow x=y$. Next observe that for every $x$, $x\le x$ which implies $\tau(x,x)=d(x,x)=0$. We have proved weak stable causality for $(X,\mathscr{T},\tau)$.

For the converse, suppose $(X,\mathscr{T},\tau)$ is weakly stably causal.
Let $x,y\in X$ be points such that $x\le y$ and $y\le x$, this means $\tau(x,y)\ge 0$ and $\tau(y,x) \ge 0$ hence  $\textrm{min}\{\tau(x,y),\tau(y,x)\} \ge 0$ from which we get $x=y$, namely the antisymmetry of $\le$. Next observe that for every $x$, $\tau(x,x)=0\ge 0$ which implies $x\le x$, thus $d(x,x)=\tau(x,x)=0$.
We conclude that $(X,\mathscr{T},\le, d)$ is weakly stably causal.

The equivalence of the conditions $d<+\infty$ and $\tau<+\infty$ is clear.
\end{proof}

A weakly stable spacetime admits the following simple characterization.
\begin{proposition}
A {\em weakly stable spacetime}  is a triple $(X,\mathscr{T},\tau)$, where   $\tau:X\times X\to \{-\infty\}\cup [0,\infty)$ is a $\mathscr{T}\times \mathscr{T}$-upper semi-continuous function such that
\begin{align*}
\forall x,y\in X& &0\le \tau(x,y)+\tau(y,x)&\Leftrightarrow x=y, \\
\forall x,y,z\in X& &\tau(x,y)+\tau(y,z)&\le \tau(x,z),
\end{align*}
again with the convention $-\infty+\infty=-\infty$.
\end{proposition}

\begin{proof}
Assume the spacetime satisfies weak stability as in Def.\ \ref{cptr}. The condition $\tau<\infty$ and second property are clear so we need only to prove the first property. If $x=y$ we know that $\tau(x,x)=0$ thus $0\le \tau(x,y)+\tau(y,x)$. For the converse, let   $x,y\in X$ be such that $0\le \tau(x,y)+\tau(y,x)$ then, due to the sum convention, both have to be different from $-\infty$, hence non-negative $\min\{\tau(x,y), \tau(y,x)\}\ge 0$, which implies $x=y$.

Assume the spacetime satisfies the given properties, we want to prove that it satisfies the properties of Def.\ \ref{cptr}.

The first property implies, choosing a pair $(x,x)$, $\tau(x,x)\ge 0$.
The reverse triangle inequality gives $2\tau(x,x)\le \tau(x,x)$ and since $0\le \tau(x,x)<+\infty$, we necessarily have $\tau(x,x)=0$.  Next, let $x,y\in X$ be such that $\min\{\tau(x,y), \tau(y,x)\}\ge 0$ then $\tau(x,y)+\tau(y,x)\ge 0$ which by the first property implies $x=y$.
\end{proof}

\begin{definition}
We say that a weakly stably causal spacetime $(X,\mathscr{T},\le, d)$ (resp.\ causal structure $(X,\mathscr{T},\le)$) is  {\em stably causal} if $(X,\mathscr{T},\le)$ is locally convex.

We say that a weakly stable spacetime $(X,\mathscr{T},\le, d)$  is a {\em stable spacetime} if $(X,\mathscr{T},\le)$ is locally convex.
\end{definition}
Thus, a spacetime is \emph{stable} if the preorder is actually an order, local convexity holds, and $d$ is finite and vanishes on the diagonal. Similarly, a causal structure is \emph{stably causal} if the preorder is actually an order and local convexity holds.

In a previous work we proved that local convexity ensures complete order regularity for locally compact, $\sigma$-compact spaces \cite[Cor.\ 2.14]{minguzzi12d}. This result will be important in what follows.

 It can be observed that in the smooth setting, under stable causality, the relation $K$ (which coincides with $J_S$) satisfies local convexity \cite[Lemma 16]{sorkin96} \cite[Lemma 5.5]{minguzzi07}\cite[Thm.\ 4.15]{minguzzi12d}.

We do not know how to express local convexity directly in terms of $\tau$.

We recall that $(X,\mathscr{T},\le)$ is a compact ordered space if $\le$ is a closed ordered space and $X$ is compact \cite{nachbin65}.
\begin{proposition}
A compact ordered space endowed with a function $d: X\times X \to [0,+\infty)$ which satisfies the reverse triangle inequality provides a stable spacetime  $(X,\mathscr{T},\le, d)$.
\end{proposition}
That is, the result does not depend on $d$  as long as the spacetime axioms are satisfied.
\begin{proof}
The reverse triangle inequality for the triple $x\le x\le x$ gives $2d(x,x)\le d(x,x)$, thus $d(x,x)=0$. Every compact ordered space is known to be completely regularly ordered hence convex \cite{nachbin65} and so locally convex.
\end{proof}

\subsection{The product trick: metricity from causality}
We now introduce the product trick, which was first defined at the level of the tangent bundle in \cite{minguzzi17,minguzzi17d}. It shows that the metrical aspects of spacetime can be reduced to causality in a space with one additional dimension. It shall provide a third definition of spacetime which is that most conveniently used in proofs.

Given a set $X$ let us denote $X^\times:=X\times \mathbb{R}$. We are interested in a topological ordered space $(X^\times,\mathscr{T}^\times,$ $\le^\downarrow)$ structured so as to respect the product structure of $X^\times$. The relation $\le^\downarrow$ will also be denoted $J^\downarrow$ (though in the manifold transition it really corresponds to $J^\downarrow_S$ of \cite[Eq.\ (4.3)]{minguzzi17}).

So consider on $X^\times$  the product topology  $\mathscr{T}^\times=\mathscr{T} \times \mathscr{T}_\mathbb{R}$
% be a topology on $X^\times$ which is the product of a topology $\mathscr{T}$ on $X$ by the topology $\mathscr{T}_\mathbb{R}$ of $\mathbb{R}$
, and let $\le^\downarrow$ be a preorder on $X^\times$ that
satisfies
\[
(*) \quad \textrm{ If } \ (p,r)\le^\downarrow (p',r') \ \textrm{ and } \ [s'-s]^+\le [r'-r]^+ \ \textrm{ then } \ (p,s)\le^\downarrow (p',s').
\]
%where $a^+:=\textrm{max}\{0,a\}$.
This condition can be better understood as follows. First, observe that a preorder that is translationally invariant,  namely for all $c\in \mathbb{R}$,
\[
(x,a)\le^\downarrow (y,b) \Rightarrow (x,a+c)\le^\downarrow (y,b+c),
\]
 can be projected to a preorder on $X$ defined by ``$x\le y$ iff there are $a,b\in \mathbb{R}$ such that $(x,a)\le (y,b)$'' (the translational invariance is used to show that this relation is transitive). Secondly, given a translationally invariant preorder $\le^\downarrow$ on $X^\times$,
 the preorder $\le^\downarrow$ contains the product of the projected preorder with the reverse canonical order\footnote{Here we are endowing the real line with the standard topology but with the reverse  canonical order. The down arrow recalls this fact. We could have used the standard order but we wanted to keep notations analogous to those of \cite{minguzzi17} which have some advantages when treating time functions as graphs, as some minus signs are not needed.} on $\mathbb{R}$ if $b\le a$ and $x\le y$ $\Rightarrow  (x,a)\le^\downarrow (y,b)$. The condition (*) is equivalent to these two conditions.
 %: translational invariance, which ensures projectability, and inclusion of the product preorder.

\begin{proposition} \label{cpl}
The condition  (*)  is equivalent to  translational invariance (which ensures projectability) and inclusion of the product preorder.
\end{proposition}

\begin{proof}
Assume  (*) and let $(x,a)\le^\downarrow (y,b)$ then $[(b+c)-(a+c)]^+\le [b-a]^+$ thus $(x,a+c)\le^\downarrow (y,b+c)$, which proves translational invariance. Since it is translationally invariant it admits a projected preorder $\le$. Suppose $x\le y$, which means that there are $a,c\in \mathbb{R}$ such that $(x,a)\le (y,c)$. Necessarily, $[c-a]^+\ge 0$, thus for $b\le a$, $[b-a]^+=0\le [c-a]^+$ and so $(x,a)\le (y,b)$, which proves that the product preorder is included in $\le^\downarrow$.

Assume translational invariance and inclusion of the product preorder. Note first that if $(p,r)\le^\downarrow (p',r')$, then $p\le p'$ and for any $r''\le r'$ we have $(p',r')\le^\downarrow (p',r'')$ because of the inclusion of the product preorder, thus by composition, $(p,r)\le^\downarrow (p',r'')$.

 Let $(p,r)\le^\downarrow (p',r')$, then $p\le p'$ and by the inclusion of the product preorder $(p,r)\le^\downarrow (p',r)$.
By translational invariance  $(p,0)\le^\downarrow (p',r'-r)$ and $(p,0)\le^\downarrow (p',0)$, and so $(p,0)\le^\downarrow (p',[r'-r]^+)$. Let $[s'-s]^+\le [r'-r]^+$ then $s'-s\le [r'-r]^+$, thus $(p,0)\le^\downarrow (p',s'-s)$, thus by translational invariance $(p,s)\le^\downarrow (p',s')$, which proves (*).
\end{proof}

Note that by (*),  $(p,r)\le^\downarrow (p',r') \Rightarrow (p,0)\le^\downarrow (p',0)$ thus the projected order is ``$x\le y$ \textrm{ iff }  $(x,0)\le (y,0)$''.

%respects the translations in the extra-dimension (equivalently, it is projectable on $X$ to a preorder $\le$), namely for all $c\in \mathbb{R}$,  $(x,a)\le^\downarrow (y,b) \Rightarrow (x,a+c)\le^\downarrow (y,b+c)$, and it contains the product preorder\footnote{Note that we are endowing the real line with the standard topology but with the reverse  canonical order. The down arrow recalls this fact. We could have used the standard one but we wanted to keep notations analogous to those of \cite{minguzzi17} which have some advantage when treating time functions as graphs, as some minus signs are not needed.} $b\le a$ and $x\le y$ $\Rightarrow  (x,a)\le^\downarrow (y,b)$.

\begin{definition} \label{pro}
A spacetime is a closed preordered space $(X^\times, \mathscr{T}^\times,\le^\downarrow)$ such that $\le^\downarrow$ satisfies (*). We say that it is {\em weakly stably causal} if $\le^\downarrow$ is  an order. If, additionally, the future (equiv.\ past) of every point does not contain an entire $\mathbb{R}$-fiber, we say that it is {\em weakly stable}.
The adjective {\em weak} is dropped if  $(X^\times, \mathscr{T}^\times,\le^\downarrow)$ is locally convex.
\end{definition}

This notion is equivalent to the previous ones. Starting from $(X,\mathscr{T},\le, d)$ just let $X^\times =X\times \mathbb{R}$, $\mathscr{T}^\times =\mathscr{T}\times \mathscr{T}_\mathbb{R}$, and
\begin{equation} \label{kcg}
(x,a)\le^\downarrow (y,b) \ \ \textrm{ iff } \ \ x\le y  \ \textrm{and} \  b\le a+d(x,y) ,
\end{equation}
equivalently, starting from $(X,\mathscr{T},\tau)$ set
\[
(x,a)\le^\downarrow (y,b) \ \ \textrm{ iff } \ \ b\le a+\tau(x,y).
\]

On the other direction, starting from $(X^\times, \mathscr{T}^\times,\le^\downarrow)$ define $\le$ as the projected preorder
\[
x\le y   \ \textrm{ iff } \   (x,0)\le^\downarrow (y,0)
\]
and set $d(x,y):=0$ if $x\nleq y$ and   otherwise
\[
d(x,y):=  \sup \{b: (x,0)\le^\downarrow (y,b) \}
\]
 (note that  $b=0$  belongs to the set thus $d\ge 0$, also since $\le^\downarrow$ is closed, $(x,0)\le^\downarrow (y,d(x,y))$). Equivalently, these two conditions read
\[
\tau(x,y):= \sup \{b: (x,0)\le^\downarrow (y,b) \},
\]
where it is understood that $\tau=-\infty$ for the empty set.

\begin{remark}
The product trick, that is, the fact that causality and metricity can be encoded in causality in one additional dimension is expressed by Eq.\ (\ref{kcg}): on the left we have the causality relation $\le^\downarrow$ on $X^\times$ while on the right we have causality $\le$ on $X$ and the Lorentzian distance $d$.
\end{remark}

\begin{proof}[Proof of the equivalence (but the last statement involving dropping `weak')] $\empty$ \\
Suppose that $(X,\mathscr{T},\le, d)$ is a spacetime, we want to prove that  $(X^\times, \mathscr{T}^\times,\le^\downarrow)$ is a closed preordered space.

Reflexivity of $\le^\downarrow$  is obvious and transitivity follows easily from the transitivity of $\le$ and the reverse triangle inequality.

The map $f: (X^\times)^2\to \mathbb{R}$, $((x,a),(y,b))\mapsto a+d(x,y)-b$ is upper semi-continuous, thus $f^{-1}([0,\infty))$ is closed. As $\pi_1\times \pi_3: ((x,a),(y,b))\mapsto (x,y)$ is continuous $\le^\downarrow=(\pi_1\times \pi_3)^{-1}(\le) \cap f^{-1}([0,\infty))$ is closed.

The  preorder $\le^\downarrow$  defined from $d$ is clearly translationally invariant,  and if $x\le y$ and  $b\le a$, then   $b\le a+d(x,y)$ thus, $(x,a)\le (y,b)$, thus $\le^\downarrow$ satisfies (*) by Prop.\ \ref{cpl}.
%Suppose that $(x,a)\nleq^\downarrow (y,b)$, then either $x\nleq y$ or $x\le y$  but $b\nleq a+d(x,y)$. In the former case, by the closure of $\le$ there are open sets $U\ni x$, $V\ni y$, such that for every $x'\in U$, $y'\in V$, $x'\nleq y'$, so that defining $U\times=U\times \mathbb{R}$, $V^\times =V\times \mathbb{R}$, we have that for $(x',a')\in U^\times$, $(y',b')\in V^\times$, $(x',a') \nleq^\downarrow (y',b')$.

For the converse: by the reflexivity of $\le^\downarrow$,  for every $x\in X$, $(x,0)\le (x,0)$ which implies $x\le x$, i.e.\ the reflexivity of $\le$.
Let $x\le y \le z$, then $(x,0)\le^\downarrow (y,d(x,y))\le^\downarrow (z,d(x,y)+d(y,z))$, so by transitivity of $\le^\downarrow$, $(x,0)\le^\downarrow (z,d(x,y)+d(y,z))$, which implies $x\le z$, namely $\le$ is transitive, and $d(x,y)+d(y,z)\le d(x,z)$, namely the reverse triangle inequality.

Suppose that $d$ is not upper semi-continuous, then we can find $x,y\in X$,  $d(x,y)<+\infty$, and $\epsilon>0$ such that for every neighborhoods $U\ni x$, $V\ni y$, and some $x'\in U$, $y'\in V$, $d(x',y') \ge d(x,y)+\epsilon$. This means that $(x',0)\le^\downarrow (y', d(x,y)+\epsilon)$, and hence by the closure of $\le^\downarrow$,  $(x,0)\le^\downarrow (y, d(x,y)+\epsilon)$, which implies $d(x,y)+\epsilon \le d(x,y)$, a contradiction that proves the upper semi-continuity of $d$.

The map $(X\times {0})^2\to X^2$, $((x,0),(y,0))\to (x,y)$, is a homeomorphism, and  $(x,0)\le^\downarrow (y,0)$ iff $x\le y$, thus it also  an order isomorphism. In order to prove the closure of $\le$ is sufficient to prove the closure of $\le^\downarrow  \cap (X\times {0})^2$ which is clear as it is an intersection of closed sets.

The condition $x \nleq y$ implies $d(x,y)=0$ follows from the definition of $d$.

Suppose that $(X,\mathscr{T},\le, d)$ is weakly stably causal.
Let $(x,a),(y,b)$ be such that $(x,a)\le^\downarrow (y,b)$ and $(x,a) \ge^\downarrow (y,b)$, then $x\le y$ and $y\le x$ which implies $x=y$ (the projected preorder is actually and order). But  the two inequalities also imply  $ b\le a+d(x,y)$ and $ a\le b+d(y,x) $, which using $x=y$, read $\vert b-a\vert\le d(x,x)=0$, thus $a=b$, namely $\le^\downarrow$ is antisymmetric, hence an order.

Suppose that $(X^\times, \mathscr{T}^\times,\le^\downarrow)$ is weakly stably causal, that is, $\le^\downarrow$ is antisymmetric. Let $x\le y$ and $y\le x$ then $(x,0)\le^\downarrow (y,0)$ and  $(y,0)\le^\downarrow (x,0)$ thus $(x,0)=(y,0)$, that is, $x=y$, namely $\le$ is antisymmetric.
Furthermore, suppose that $d(x,x)>0$ for some $x$, then $(x,0)\le^\downarrow (x,d(x,x))$ and $(x,d(x,x)) \le^\downarrow (x,0)$, a contradiction.

%The equivalence of stable casuality in the two formulation is clear because one property is the antisymmetry of the quotient relation $\le$ which is formulated in equivalent ways. The preorded induce on the fiber with projection $x$ is given by $(x,a)\le^\downarrow (x,b)$, that is $b\le a +d(x,x)$ which coincides with $b\le a$ (the reverse canonical order of the real line) iff $d(x,x)=0$.

% also for $(X^\times, \mathscr{T}^\times,\le^\downarrow)$ is expressed as the antisymmetry of the quotient relation $\le$.

%Let us show that the stability of $(X,\mathscr{T},\le, d)$  implies stability  of  $(X^\times, \mathscr{T}^\times,\le^\downarrow)$.

%Suppose $(X,\mathscr{T},\le, d)$  is stable.
%Let $(x,a),(y,b)$ be such that $(x,a)\le^\downarrow (y,b)$ and $(x,a) \ge^\downarrow (y,b)$, then $x\le y$ and $y\le x$ which implies $x=y$ (the quotient preorder is actually and order). But  the two inequalities also imply  $ b\le a+d(x,y)$ and $ a\le b+d(y,x) $, which using $x=y$, read $\vert b-a\vert\le d(x,x)$.  But $d$ is finite thus $d(x,x)=0$  and we conclude $a=b$, namely $\le^\downarrow$ is antisymmetric, hence an order. This prove that $(X^\times, \mathscr{T}^\times,\le^\downarrow)$ is stable.

%Let us show that the stability of $(X^\times, \mathscr{T}^\times,\le^\downarrow)$ implies stability  of  $(X,\mathscr{T},\le, d)$.

%Suppose $(X^\times, \mathscr{T}^\times,\le^\downarrow)$  is stable. We already know that $\le$ is antisymmetric.
As for stability, observe that
the future of $(x,a)$ reads
\[
(J^\downarrow)^+((x,a))=\{(y,b): x \le y \textrm { and } b \le a + d(x,y)\}
\]
so it does not contain a whole $\mathbb{R}$-fiber iff $d(x,y)$ is finite for every $y$. A similar conclusion holds studying the past of a point
\[
(J^\downarrow)^-((y,b))=\{(x,a): x \le y \textrm { and } b \le a + d(x,y)\}.
\]
\end{proof}

%The fact that the local convexity in $X^\times$ can be established starting from analogous properties in $X$ is non-trivial and will play an important role.
We end the proof of the equivalence of Definition \ref{pro} with the previous definitions for spacetime, with the following result.

\begin{proposition} \label{otyu}
Let $(X,\mathscr{T},\le,d)$ be a  weakly stably causal spacetime. Local convexity holds for $(X,\mathscr{T},\le)$ iff it holds for $(X^\times, \mathscr{T}^\times,$ $\le^\downarrow)$.
\end{proposition}

\begin{proof}
$\Rightarrow$. By translational invariance it is sufficient to prove local convexity at $(x,0)\in X^\times$, $x\in X$. Let $O\times (-\epsilon, \epsilon)$ be a product open neighborhood of  $(x,0)$. We want to prove that there is a convex neighborhood $W$ of  $(x,0)$ inside this product open set. Since $d(x,x)=0$, and $d$ is upper semi-continuous, we can assume, without loss of generality, that $O$ is such that $d\vert_{O\times O} < \epsilon/2$. Since $X$ is locally convex there is a convex neighborhood $U\ni x$, $U\subset O$. Let $R=U\times (-\epsilon/2, \epsilon/2)$ and $W=i^\downarrow(R)\cap d^\downarrow(R)$. Let us prove that this convex neighborhood of $(x,0)$ satisfies  $W\subset O\times (-\epsilon, \epsilon)$. Indeed, if $(p,a), (r,c)\in R$ and $(p,a)\le^\downarrow (q,b) \le^\downarrow (r,c)$ then $p\le q \le r$ thus, as $p,r\in U$ which is convex, $q\in U$. Moreover, $b\le a  + d(p,q)$ and $c\le b+d(q,r)$, but $d(p,q),d(q,r)< \epsilon/2$, thus $b<a+\epsilon/2<\epsilon$ and $b> c-\epsilon/2>-\epsilon$, thus $(q,b)\in   O\times (-\epsilon, \epsilon)$.

$\Leftarrow$. If  $(X^\times, \mathscr{T}^\times,$ $\le^\downarrow)$ is locally convex, in order to prove local convexity at $x\in X$, let $O\ni x$ be an open neighborhood, then $O\times (-1,1)$ is an open neighborhood of $(x,0)$ and so there is a convex neighborhood $W$ of $(x,0)$ contained in  $O\times (-1,1)$ . Let $U\subset O$ be the projection of the elements of $W$ that have real coordinate equal to zero. The neighborhood $W$ contains  a product open neighborhood of $(x,0)$, thus the set $U$ is  a neighborhood of $x$. The set $U$ is convex because if $p,r\in U$ and $q$ is such that $p\le q\le r$, then $(p,0)\le^\downarrow (q,0)\le^\downarrow (r,0)$ and as $(p,0),(r,0)\in W$, we get by its convexity, $(q,0)\in W$ and hence $q\in U$, namely $U$ is convex.
\end{proof}

\subsection{Global hyperbolicity}

In this section we define the notion of globally hyperbolic spacetime.

\begin{proposition} \label{wekp}
For a weakly stably causal spacetime $(X, \mathscr{T}, \le, d)$, equivalently $(X^\times,$ $\mathscr{T}^\times,$ $\le^\downarrow)$, the following properties are equivalent
\begin{itemize}
%\item[(a)] for every $p,q\in X$, the diamond $J(p,q)$ is compact,
\item[(i)]  for every $p,q\in X$, $J(p,q)$ is compact, and $d$ is finite (hence the spacetime is weakly stable),
\item[(ii)] for every $P,Q\in X^\times$, the diamond $J^\downarrow(P,Q):=(J^\downarrow)^+(P)\cap (J^\downarrow)^-(Q)$ is compact.
\end{itemize}
\end{proposition}

\begin{proof}
(i) $\Rightarrow$ (ii). Since $J^\downarrow$ is closed, $J^\downarrow(P,Q)$ is closed. The projection on $X$ is contained in  $J(p,q)$, that to $\mathbb{R}$ is contained in the compact interval $[\pi_\mathbb{R}(Q)-d(p,q), \pi_\mathbb{R}(P)+d(p,q)]$. Thus $J^\downarrow(P,Q)$ is contained in the product of compact sets which is compact.

(ii) $\Rightarrow$ (i).
%Since  $(X,\mathscr{T},\le,d)$ is weakly stably causal, we need only to
Let us prove that $d$ is finite. Suppose $d(p,q)=+\infty$ for some $p,q\in X$, which implies $p\le q$. Let $P=(p,0)$, $Q=(q,0)$ then the points $R=(q,r)$, for $r>0$ are included in the diamond $J^\downarrow(P,Q)$, which is thus non-compact as its projection to $\mathbb{R}$ is non-compact, a contradiction.

Let $r\in J(p,q)$ then $(r,0)\in J^\downarrow(P,Q)$, where $P=(p,0)$, $Q=(q,0)$, which proves that $J(p,q)$ is contained in the projection of  $J^\downarrow(P,Q)$ which is compact. But $J(p,q)$ is closed because $\le$ is closed, hence compact.
%The set $J(p,q)$ is closed because $\le$ is closed. It is also contained in
%The compactness of $J(p,q)$ follows projecting the compact set $J^\downarrow(P,Q)$ to $X$.
\end{proof}

We have also the following variant where $J(C):=J^+(C)\cap J^-(C)$ and similarly for $J^\downarrow$.

\begin{proposition} \label{strp}
For a weakly stably causal spacetime $(X, \mathscr{T}, \le, d)$, equivalently $(X^\times,$ $\mathscr{T}^\times,$ $\le^\downarrow)$, the following properties are equivalent
\begin{itemize}
\item[(i)]  for every compact set $C$, $J(C)$  is compact, and $d$ is finite (hence the spacetime is weakly stable),
\item[(ii)] for every compact set $K$, $J^\downarrow(K)$ is compact.
\end{itemize}
\end{proposition}

\begin{proof}
(i) $\Rightarrow$ (ii). Since $J^\downarrow$ is closed, $J^\downarrow(K)$ is closed and by reflexivity of $J^\downarrow$, it contains $K$. Let $C:=\pi_X(K)$ and $D:=\sup_{C\times C} d$, which is attained by the upper semi-continuity of $d$. Since $d$ is finite, $D$ is finite.

We have  $\pi_X(J^\downarrow(K))\subset J(C)$ and the projection to $\mathbb{R}$ is contained in the compact interval $[\inf \pi_\mathbb{R}(K)-D, \sup \pi_\mathbb{R}(P)+D]$, where the inf and sup in this expression are finite by the compactness of $K$. Thus $J^\downarrow(K)$ is contained in the product of compact sets which is compact.

(ii) $\Rightarrow$ (i).
%Since  $(X,\mathscr{T},\le,d)$ is weakly stably causal, we need only to
Let us prove that $d$ is finite. Suppose $d(p,q)=+\infty$ for some $p,q\in X$, which implies $p\le q$ and $p\ne q$ and hence $q \nleq p$. Let $P=(p,0)$, $Q=(q,0)$ then the points $R=(q,r)$, for $r>0$ are included in the diamond $J^\downarrow(P,Q)=J^\downarrow(K)$, $K=\{P,Q\}$, which is thus non-compact as its projection to $\mathbb{R}$ is non-compact, a contradiction.

Let $C$ be a compact set and let $K=C\times \{0\}$. Let $r\in J(C)$ then there are $p,q\in C$, $p\le r\le q$. Note that  $(r,0)\in J^\downarrow(P,Q)$, where $P=(p,0)$, $Q=(q,0)$ which proves that $J(C)$ is contained in the projection of  $J^\downarrow(K)$ which is compact. But $J(C)$ is closed because $\le$ is closed, hence compact.
%The set $J(p,q)$ is closed because $\le$ is closed. It is also contained in
%The compactness of $J(p,q)$ follows projecting the compact set $J^\downarrow(P,Q)$ to $X$.
\end{proof}

The property of Prop.\ \ref{strp} implies that of Prop.\ \ref{wekp}, it is sufficient to use the antisymmetry of $\le$ and take $C=\{p,q\}$, noting that if $J(p,q)\ne \emptyset$ then $p\le q$ and $J(p,q)=J(\{p,q\})$.

The equivalence of these two variants generally requires further assumptions (their equivalence and use to define global hyperbolicity in the smooth manifold analogy is essentially the equivalence in \cite[Prop.\ 2.21]{minguzzi17} which uses a {\em proper cone structure} assumption).

In the theory of closed ordered spaces the property of
{\em $k$-preservation} \cite{minguzzi12d}---for each compact subset $K\subset X$, $i(K)\cap d(K)$ is compact---has been proposed to define global hyperbolicity. It is clearly equivalent to the property of Prop.\ \ref{strp} for weakly stable spacetimes (as for them $d$ is finite).

For definiteness, in the context of this work, we give the following definition

\begin{definition}
A stable $k$-preserving spacetime $(X, \mathscr{T}, \le, d)$ is said to be {\em globally hyperbolic}.
\end{definition}
Under local compactness and $\sigma$-compactness, the definition can be improved, as $k$-preservation already implies local convexity \cite[Thm.\ 3.3]{minguzzi12d}.

\begin{example}
Let $(X,d)$ be a countably generated Lorentzian metric space in the sense of \cite{minguzzi24b}. This means the following conditions hold:
\begin{itemize}
\item[(i)] there is a function (Lorentzian distance) $d:X\times X\to [0,\infty)$ that satisfies the reverse triangle inequality over chronologically related triples ($I:=\{d>0\}$ is the chronological relation),
\item[(ii)] there is a topology $\mathscr{T}$ such that $d$ is continuous and $\overline{I(p,q)}$ is compact for every $p,q\in X$,
\item[(iii)] for any two distinct points $p,q\in X$ there is a third point $r$ such that $d(p,r)\ne d(q,r)$ or $d(r,p) \ne d(r,q)$,
\item[(iv)] there is a countable set $\mathcal{S}\subset X$ such that for every $q\in X$ there are $p,r\in \mathcal{S}$ such that $q\in I^+(p)\cap I^-(r)$.
\end{itemize}
The topology $\mathscr{T}$ is actually unique and called the {\em Lorentzian metric space topology}. Similarly, there is a natural closed order $J$ (uniquely determined from $d$), such that $I\subset J$, and the reverse triangle inequality extends to $J$-related triples. It is known \cite[Thm.\ 4.6]{minguzzi24b} that for every closed order $I\subset K\subset J$ and  every compact set $C$, the set  $K(C)$ is compact.  As a consequence $(X,\mathscr{T}, K, d)$ is a second-countable locally compact $\sigma$-compact (actually Polish) $k$-preserving stable spacetime (in the sense of this work), in particular the following main Thm.\ \ref{mai} on the representation of spacetime via continuous rushing functions applies to it (see  \cite[Prop.\ 3.20]{minguzzi24b} for the topological properties).
Note that if there is a gap between $\overline{I}$ and $J$, the continuous $d$-rushing functions might depend on the chosen $K$ though $d$ does not change. There is no gap if for every $p\in X$, $p\in \overline{I^\pm(p)}$ \cite{minguzzi22,minguzzi24b}.

%Smooth globally hyperbolic spacetimes $(M,g)$ belong to the above class.
Crucially, smooth globally hyperbolic spacetimes $(M,g)$ belong to the above class, which shows our abstract definition encompasses the standard classical setting.
In any case, for them more refined representation theorems are available \cite{minguzzi17}.

The same considerations apply to bounded Lorentzian metric spaces \cite{minguzzi22} but for them there is no need to assume the countably generated property (property (iv) above) as a Polish topological property follows anyway \cite[Thm.\ 1.10]{minguzzi22}.
\end{example}

\subsection{Order and product of spacetimes}

Given two spacetimes $M_1=(X_1,\mathscr{T}_1,\le_1, d_1)$ and $M_2=(X_2,\mathscr{T}_2,\le_2, d_2)$, let us write\footnote{We could generalize this definition dropping $X_1\subset X_2$ by requiring that there is a  map $\phi_{12}: X_1\to X_2$ which is injective, continuous, and preserves the distance.}: $X_1 \preceq X_2$ if  $X_1 \subset X_2$; $\mathscr{T}_1 \preceq \mathscr{T}_2$ if $\mathscr{T}_1 \supset \mathscr{T}_2$; $\le_1\preceq \le_2$ if $\le_1 \subset \le_2$; and $d_1\preceq d_2$ if $d_1(x,y)\le d_2(x,y)$ for every $x,y\in X_1$.
Observe that the empty set, the discrete topology, the discrete order, and the vanishing  Lorentzian distance, $d:=0$, are all $\preceq$-lower bounds on the respective categories. We can admit $X=\emptyset$ as a possible set for a spacetime, which is necessarily $(\emptyset,\{\emptyset\}, \emptyset, 0)$.
 Let us also write $M_1\preceq M_2$ if $X_1 \preceq X_2$, $\mathscr{T}_1 \preceq \mathscr{T}_2$,  $\le_1\preceq \le_2$ and  $d_1\preceq d_2$. It is clear that $\preceq$ is an order.

\begin{proposition}
Let $M_\alpha, \alpha\in A$, be a family of spacetimes.
% and suppose that $\cap_{\alpha \in A} X_\alpha$ is non-empty.
Then the quadruple $M:=(X,\mathscr{T},\le, d)$ with
\[
X:=\cap_{\alpha \in A} X_\alpha, \quad \mathscr{T}:=\sup_{\alpha\in A}   \mathscr{T}_\alpha, \quad  \le:=\cap_{\alpha\in A}   \le_\alpha, \quad d:=\inf_{\alpha\in A} d_\alpha ,
\]
is a spacetime. It is  the largest lower bound (i.e.\ the infimum) of $\{M_\alpha\}$ for the order $\preceq$. If one of the $M_\alpha$ is weakly stably causal (resp.\ weakly stable) then so is $M$. If all the $M_\alpha$ are stably causal (resp.\  stable) then so is $M$.
\end{proposition}

Note that every topology $\mathscr{T}_\alpha$, preorder $\le_\alpha$, and Lorentzian distance $d_\alpha$, can be restricted to the smaller set $X$. In the equation in display we omitted the restriction operation.

Inspection of the proof shows that if the family $\{M_\alpha\}$ is totally ordered, then in the last statement we can replace ``If all the $M_\alpha$ are'' by ``If one the $M_\alpha$ is''.

\begin{proof}
Let us prove that it is indeed a spacetime.
%Every topology $\mathscr{T}_\alpha$, preorder $\le_\alpha$, and Lorentzian distance $d_\alpha$, can be restricted to the smaller set $X$.
The intersection of preorders is easily show to be a preorder, thus $\le$ is a preorder. Observe that $\le_\alpha$ is closed in the $\mathscr{T}_\alpha\times \mathscr{T}_\alpha$ topology, thus $\le_\alpha\!\!\vert_{X\times X}$ is closed in the $\mathscr{T}_\alpha\times \mathscr{T}_\alpha\vert_{X\times X}$ topology, hence in the finer $\mathscr{T}\times \mathscr{T}\vert_{X\times X}$ topology, thus the intersection $\cap_{\alpha\in A} \le_\alpha\!\!\vert_{X\times X}$ is closed in the $\mathscr{T}\times \mathscr{T}\vert_{X\times X}$ topology.
Let $x,y\in X$, $x\nleq y$ then for some $\alpha$, $x\nleq_\alpha y$ which implies $d_\alpha(x,y)=0$ and hence $d(x,y)=0$.
Let $x,y,z\in X$ be such that $x\le y\le z$ then $x\le_\alpha y \le_\alpha z$ for each $\alpha$ and
\[
d_\alpha(x,y)+d_\alpha(y,z)\le d_\alpha(x,z)
\]
for each $\alpha$. Thus
\[
\inf_{\alpha\in A} d_\alpha(x,y)+\inf_{\alpha\in A} d_\alpha(y,z)\le \inf_{\alpha\in A}[d_\alpha(x,y)+d_\alpha(y,z)]\le \inf_{\alpha\in A} d_\alpha(x,z)
\]
which proves that $d$ satisfies the reverse triangle inequality. Finally, each $d_\alpha\vert_{X\times X}$ is $\mathscr{T}_\alpha\times \mathscr{T}_\alpha\vert_{X\times X}$-upper semi-con\-ti\-nu\-ous and so $\mathscr{T}\times \mathscr{T}\vert_{X\times X}$-upper semi-con\-ti\-nu\-ous. The infimum of a family of upper semi-continuous functions is upper semi-continuous, thus $d$ is  $\mathscr{T}\times \mathscr{T}\vert_{X\times X}$-upper semi-continuous. Clearly, $M\preceq M_\alpha$ for each $\alpha$, thus $M$ is a lower bound for $\{M_\alpha\}$. If $N$ is a lower bound for $\{M_\alpha\}$, then $X_N\subset X_\alpha$, $\mathscr{T}_N\supset \mathscr{T}_\alpha$, $\le_N\subset \le_\alpha$, $d_N\le d_\alpha$, for every $\alpha$, which implies $X_N\subset X$, $\mathscr{T}_N\supset \mathscr{T}$, $\le_N\subset \le$, $d_N\le d$, namely $N\preceq M$. We conclude that $M$ is the largest lower bound.

The last two statements are straightforward, the only possible difficulty is the proof that if all $M_\alpha$ are locally convex then so is $M$. Indeed, let $x\in M$ and let $O\ni x$, $O\in \mathscr{T}$, then we can find some $O_{i}\in \mathscr{T}_{\alpha_i}$, for $i=1,\ldots, k$, with $k$ finite, such that $\cap_i O_{i} \subset O$. By local convexity of $X_{\alpha_i}$ we can find a $\le_{\alpha_i}$-convex neighborhood $C_i\ni x$ such that  $C_i \subset O_i$. Let $C=\cap_i C_i$, then $C\subset O$ is $\le$-convex.
\end{proof}

\begin{proposition} \label{cpfp}
Let $M_\alpha$, $\alpha\in A$, be a family of spacetimes $M_\alpha:=(X_\alpha,$ $\mathscr{T}_\alpha,\le_\alpha, d_\alpha)$. Then the quadruple $M:=(X,\mathscr{T},\le, d)$ in which $X=\Pi_{\alpha \in A} X_\alpha$ is the Cartesian product, $\mathscr{T}$ is the product topology, $\le$ is the product preorder, and
\[
 d:=\inf_{\alpha\in A} d_\alpha \circ (\pi_\alpha\times \pi_\alpha),
\]
is  a spacetime which we call the {\em product  spacetime}. If all  $M_\alpha$ are weakly stably causal (resp.\ weakly stable) then so is $M$.
If all  $M_\alpha$ are stably causal (resp.\ stable, $k$-preserving) then so is $M$.
%It is  the largest lower bound (i.e.\ the infimum) of $\{M_\alpha\}$ for the order $\preceq$.
\end{proposition}

\begin{proof}
The composition $f\circ g$ with $g$ continuous and $f$ upper semi-continuous functions is upper semi-continuous. Thus, as all projections $\pi_\alpha: X \to X_\alpha$ are continuous in the product topology, $d_\alpha \circ (\pi_\alpha\times \pi_\alpha)$ is upper semi-continuous. The infimum of an arbitrary family of upper semi-continuous functions is upper semi-continuous, thus $d$ is upper semi-continuous. If $x\nleq y$, then for some $\alpha$, $x_\alpha \nleq y_\alpha$, which implies $d_\alpha(x_\alpha, y_\alpha)=0$, and hence  $d(x,y)=0$. For every triple,  $x\le y\le z$, we have for each $\alpha$,  $x_\alpha\le y_\alpha\le z_\alpha$, thus
\[
d_\alpha(x_\alpha,y_\alpha)+d_\alpha(y_\alpha,z_\alpha)\le d_\alpha(x_\alpha,z_\alpha).
\]
hence
\begin{align*}
d(x,y)+d(y,z)&=\inf_{\alpha}d_\alpha(x_\alpha,y_\alpha)+\inf_\alpha d_\alpha(y_\alpha,z_\alpha) \\
&\le  \inf_{\alpha} [d_\alpha(x_\alpha,y_\alpha)+ d_\alpha(y_\alpha,z_\alpha)]\le \inf_\alpha d_\alpha(x_\alpha,z_\alpha)=d(x,z).
\end{align*}
Finally, $\le=\cap_\alpha (\pi_\alpha\times \pi_\alpha)^{-1}(\le_\alpha)$ which being the intersection of closed sets is closed.

The last two statements are straightforward, the only possible difficulty is the proof that if all $M_\alpha$ are locally convex (resp.\ $k$-preserving) then so is $M$. Let us start with local convexity. Let $x\in M$ and let $O\ni x$, $O\in \mathscr{T}$, then we can find some $O_{i}\in \mathscr{T}_{\alpha_i}$, for $i=1,\ldots, k$, with $k$ finite, such that $\cap_i \pi_{X_{\alpha_i}}^{-1}( O_{i})\subset O$. By local convexity of $X_{\alpha_i}$ we can find a $\le_{\alpha_i}$-convex neighborhood $C_i\ni \pi_{\alpha_i}(x)$ such that  $C_i \subset O_i$. Let $C=\cap_i \pi_{X_{\alpha_i}}^{-1}(C_{i})$, then $C\subset O$ is $\le$-convex and it is also a neighborhood of $x$.

Let us consider the $k$-preserving property. Let $K$ be a compact subset of $X$, then $K_\alpha=\pi_{X_\alpha}(K)$ is compact. The set $c(K)=i(K)\cap d(K)$  is closed because $M$ is a closed ordered space. Let $y\in c(K)$ then there are $x,z\in K$, $x\le y \le z$, which implies $x_\alpha \le_\alpha y_\alpha \le_\alpha z_\alpha$, but $x_\alpha, z_\alpha \in K_\alpha$, thus $y_\alpha \in c(K_\alpha)$ which proves that $y\in \Pi_\alpha c(K_\alpha)$ and so $c(K)\subset \Pi_\alpha c(K_\alpha)$. Since the product is compact, we conclude that $c(K)$ is compact.
\end{proof}

\begin{example}
The simplest example of spacetime is the real line $\mathbb{R}$, with $\mathscr{T}$ and $\le$ the usual topology and order, and $d(x,y):=\max\{0, y-x\}$. It is locally compact, $\sigma$-compact, stable and $k$-preserving, hence globally hyperbolic.
We shall be particularly interested in the spacetime $\mathbb{R}^A$ obtained by products of the real line which, by the previous results, is stable and $k$-preserving. Its Lorentzian distance is $d(x,y)= \max\{0, \inf_\alpha[ y_\alpha-x_\alpha]\}$.
\end{example}

%\begin{example}
%A variant is the extended real line $\bar{\mathbb{R}}=\{\pm\infty\}\cup\mathbb{R}$, with $\mathscr{T}$ and $\le$ the usual topology and order, and $d(x,y):=\max\{0, y-x\}$ and, whenever an operation is indeterminate, $d:=+\infty$. This ensures that $d$ is upper semi-continuous.
%\end{example}

The following result is well known. We include the proof for completeness. Actually, we shall only make use of if for finite sequences.

\begin{lemma} \label{coop}
For $p\in (0,1]$ let for any non-negative real sequence $x=\{x_k, k\in \mathbb{N}\}$, $\Vert x\Vert_p:=\{\sum_k  x_k^p\}^{1/p}$, then for any two non-negative real sequences $x=\{x_k\}, y=\{y_k\}$, $\Vert x+y\Vert_p\ge \Vert x\Vert_p+\Vert y\Vert_p$.
\end{lemma}

It is clear that the result holds also if the elements of the sequences are allowed to assume the value $+\infty$ (if this happens we have actually an equality $+\infty=+\infty$).

\begin{proof}
If $\Vert x\Vert_p=0$ or $\Vert y\Vert_p=0$ then the result is trivial as at least one among $x$ and $y$ vanishes. The claim is also clear if $\Vert x\Vert_p=+\infty$ or $\Vert y\Vert_p=+\infty$ as $\Vert x+y\Vert_p$ is no smaller than  them. Let us assume, without loss of generality, that $\Vert x\Vert_p+\Vert y\Vert_p$ is finite. Since we have positive homogeneity $\Vert c x \Vert_p=c \Vert x \Vert_p$, for $c>0$, we can limit ourselves to the case $\Vert x \Vert_p=1-\lambda$, $\Vert y \Vert_p=\lambda$, for $\lambda\in (0,1)$. So defining $X=x/(1-\lambda)$ and $Y=y/\lambda$ we need only to prove the concavity property  $\Vert (1-\lambda)X+\lambda Y\Vert_p\ge 1$ for every $X$ and $Y$ such that  $\Vert X\Vert_p=1=\Vert Y\Vert_p$, and $\lambda\in (0,1)$. Now, for every coordinate $i$, due to the concavity of $f(x)=x^p$ on $[0,+\infty)$,  $[(1-\lambda)X_i+\lambda Y_i]^p\ge (1-\lambda)X_i^p+\lambda Y_i^p$. Summing over $i$, $\Vert (1-\lambda) X + \lambda Y\Vert_p^p \ge (1-\lambda) \Vert X\Vert_p^p+\lambda \Vert Y\Vert_p^p=1$.
\end{proof}

\begin{proposition}
Let $M_k$, $k \in A$, be a finite family of spacetimes $M_k:=(X_k,$ $\mathscr{T}_k,\le_k, d_k)$. Then the quadruple $M:=(X,\mathscr{T},\le, d)$ in which $X=\Pi_{k\in A} X_k$ is the Cartesian product, $\mathscr{T}$ is the product topology, $\le$ is the product preorder, and
\[
 d\vert_{\le}:= \Big\{\sum_k [d_k \circ (\pi_k\times \pi_k)]^p \Big\}^{\! 1/p}, \quad  d\vert_{\nleq}:=0
\]
for $p\in (0,1]$ is  a spacetime which we call the {\em $p$-product  spacetime}. If all  $M_k$ are weakly stably causal (resp.\ weakly stable) then so is $M$.
If all  $M_k$ are stably causal (resp.\ stable, $k$-preserving) then so is $M$.
%It is  the largest lower bound (i.e.\ the infimum) of $\{M_\alpha\}$ for the order $\preceq$.
\end{proposition}

Note that, for given factors, these $p$-product spacetimes and the product spacetime of  Prop.\ \ref{cpfp} share the same causal structure. The request of the upper semi-continuity of $d$ is really the only condition that forces us to  consider just  finite, rather than countable, $p$-products.

\begin{proof}
The projections $\pi_k: X \to X_k$  are continuous, thus $\le=\cap_k (\pi_k\circ \pi_k)^{-1}(\le_k)$, being the intersection of closed sets, is closed.

Again by the continuity of the projections,  $d_k \circ (\pi_k\times \pi_k)$ is upper semi-continuous. Let $x\nleq y$ then as $\le$ is closed there is a neighborhood $O$ of $(x,y)$ such that for every $(x',y')\in O$, $x'\nleq y'$, and thus $d(x',y')=0$. This proves that $d$ is upper semi-continuous outside $\le$. Suppose $x\le y$ then $x_i\le y_i$ for every $i$. If for some  $i$, $d_i(x_i,y_i)=+\infty$ then $d(x,y)=+\infty$ and the upper semi-continuity is clear. So suppose that for every $i$, $d_i(x_i,y_i)<+\infty$.  Let $\epsilon>0$, we can find $O_i\ni (x_i,y_i)$ such that $d_i\vert_{O_i} <d(x_i,y_i)+\epsilon$, but then for $O=\Pi_i O_i$
\[
d\vert_O\le \Big\{\sum_k [d_k(x_k,y_k) +\epsilon]^p \Big\}^{\! 1/p}
\]
where the right-hand side converges to $d(x,y)$ for $\epsilon \to 0$ (this step uses the finiteness of the sum), thus $d$ is upper semi-continuous at $(x,y)$.

For every triple,  $x\le y\le z$, we have for each $k$,  $x_k\le y_k\le z_k$, thus
\[
d_k(x_k,y_k)+d_k(y_k,z_k)\le d_k(x_k,z_k).
\]
Thus $\sum_k d_k(x_k,z_k)^p \ge  \sum_k [d_k(x_k,y_k)+d_k(y_k,z_k)]^p$ and by  Lemma \ref{coop}
\[
d(x,z)=\Big(\sum_k d_k(x_k,z_k)^p\Big)^{\! 1/p} \ge d(x,y)+d(y,z).
\]

If every $d_i$ vanishes on the diagonal, the same holds for $d$. If $d_i$ are all finite, the same holds for $d$.
The last two statements depend then only on the causal structures and so the proof is the same as that of Prop.\ \ref{cpfp}.
\end{proof}

\section{Spacetimes from functional spaces}

%The next result proves that stable spacetimes are quite natural objects as they  follow from just a family of real functions over a set. It is inspired by \cite[Thm.\ 4.6]{minguzzi17} which was restricted to manifolds. Note that here the steep functions are replaced by the rushing functions. For the relationship between steep and rushing functions on a manifold, see \cite[Thm.\ 1.28]{minguzzi18b}.

The next result proves that stable spacetimes are  natural objects that can be constructed  from just a family of real-valued functions over a set. It is inspired by \cite[Thm.\ 4.6]{minguzzi17} which was restricted to manifolds. Note that here the $F$-steep functions are replaced by the rushing functions \cite[Def.\ 1.26]{minguzzi18b}. On a spacetime $(X,\mathscr{T},\le, d)$ a function $X\to \mathbb{R}$ is {\em rushing} if
\[
\forall x,y \in X, \qquad x\le y  \ \Rightarrow \ f(x)+d(x,y) \le f(y)
\]
 %for every pair $x\le y$, $f(y)\ge f(x)+d(x,y)$ (
 equivalently, for every $x,y\in X$, $f(x)+\tau(x,y) \le f(y)$. Intuitively, a rushing function,  interpreted as a time, runs faster than proper time.
For the relationship between steep and rushing functions on a manifold, see \cite[Thm.\ 1.28]{minguzzi18b}.

\begin{theorem} \label{ckkr}
Let $\mathcal{F}$ be a family of real valued functions on $X$ that distinguishes points: $f(x)=f(y)$ for every $f\in \mathcal{F}$ implies $x=y$. Then
\begin{itemize}
\item[(i)] The initial topology $\mathscr{T}$ generated by $\mathcal{F}$ is Tychonoff (hence Hausdorff). If $\mathcal{F}$ is countable then the initial topology is second countable, hence metrizable.
\item[(ii)] The  relation defined by
\begin{equation} \label{ord}
\le:=\{(x,y): f(x)\le f(y), \ \forall f\in \mathcal{F}\}
\end{equation}
 is an order which is closed with respect to the product of the initial topology of point (i).
 %(and hence with respect to any finer topology, that is any topology that makes all the elements of $\mathcal{F}$ continuous).
\item[(iii)]    The function $d: X^2\to [0,\infty)$ defined by
\begin{equation} \label{dis}
d(x,y):= \max\{0, \inf_{\mathcal{F}} [f(y)-f(x)]\}
\end{equation}
is upper semi-continuous (with respect to the product of the initial topology of point (i)); satisfies $x\nleq y \Rightarrow d(x,y)=0$; for every $x\in X$, $d(x,x)=0$; and it  satisfies the reverse triangle inequality with respect to the relation $\le$ of point (ii).
\item[(iv)]
 Suppose that for every $\lambda \ge 1$, $\lambda \mathcal{F}\subset \mathcal{F}$, then we have the direct expressions
\begin{equation} \label{dis2}
\tau(x,y)= \inf_{\mathcal{F}} [f(y)-f(x)],
\end{equation}
and
\begin{equation}
(x,a)\le^\downarrow (y,b) \ \ \textrm{ iff } \ \ \forall f\in \mathcal{F}, \ \ f(x)-a\le f(y)-b.
\end{equation}
\end{itemize}
As a consequence, $(X,\mathscr{T},\le, d)$ is a stable spacetime and every element $f\in \mathcal{F}$ is $\mathscr{T}$-continuous, $\le$-isotone and $d$-rushing. Moreover, $(X,\mathscr{T},\le)$ is a completely regularly ordered space (causal structure) hence convex.
% \cite[Def.\ 1.26]{minguzzi18b} in the sense that, for every pair $x\le y$, $f(y)\ge f(x)+d(x,y)$ (equivalently, for every $x,y\in X$, $f(y)\ge f(x)+\tau(x,y)$).
\end{theorem}

Without the condition on the distinction of points, the topology would not be $T_0$ and  $\le$ would not be antisymmetric (it would be a preorder). In principle one could work with indistinguishable points but then the Alexandrov quotient would need to be taken.
%The order in (ii) is sometimes denoted $\le_{\mathcal{F}}$.

The initial topology induced by a family of functions $\mathcal{F}$ is also denoted $\mathscr{T}_{\mathcal{F}}$, a preorder given by formula (\ref{ord}) is also denoted $\le_{\mathcal{F}}$, and a Lorentzian distance given by formula (\ref{dis}) is also denoted $d_{\mathcal{F}}$. Thus for any family of functions ${\mathcal{F}}$  that distinguishes points, we have  a stable spacetime $(X,\mathscr{T}_{\mathcal{F}},\le_{\mathcal{F}}, d_{\mathcal{F}})$.

\begin{proof}
Let $x\ne y$ then there is $f\in \mathcal{F}$ such that $f(x)\ne f(y)$. Let $r=[f(x)+ f(y)]/2$, then the initial topology open sets   $f^{-1}((-\infty,r))$, $f^{-1}((r,+\infty))$ are disjoint and each of them contains precisely one of the two points. Thus $\mathscr{T}$ is Hausdorff. %hence $T_1$.

%As $\mathcal{F}$ distinguishes points, the initial topology is Hausdorff, hence $T_1$.
The elements of  $\mathcal{F}$  are continuous in the initial topology $\mathscr{T}$, thus the topology $\mathscr{T}$ is generated by a family of continuous functions which proves that $(X,\mathscr{T})$ is completely regular hence Tychonoff. If  $\mathcal{F}=\{f_k\}_{k\in \mathbb{N}}$ then a subbasis for $\mathscr{T}$ is given by the finite intersections of  sets of the form $f^{-1}_k((a, +\infty))$, $f^{-1}_k((-\infty, b))$, with $a,b\in \mathbb{Q}$. Since the subbasis is  countable so is the basis, that is, under countability of $\mathcal{F}$ the topology $\mathscr{T}$ is second-countable. Every Tychonoff space is regular, and every regular second countable space is metrizable.

We have $\le = \cap_{f\in \mathcal{F}} (f\times f)^{-1}(\le_\mathbb{R})$.
% thus if $\mathscr{T}'$ is a topology with respect to which each $f\in \mathcal{F}$ is continuous,
The sets on the right-hand side are $\mathscr{T}\times \mathscr{T}$-closed (as $\le_\mathbb{R}$ is closed, and $f\times f$ is $\mathscr{T}\times \mathscr{T}$-continuous) and so is their intersection.

The fact that $\le$ is an order is clear, the antisymmetry following from the fact that $\mathcal{F}$ distinguishes points.

The function $\inf_{\mathcal{F}} [f(y)-f(x)]$ being the infimum of a family of upper semi-continuous functions is upper semi-continuous. Now, the maximum of a finite family of upper semi-continuous functions (such as  $\max\{0, \inf_{\mathcal{F}} [f(y)-f(x)]\}$) is upper semi-continuous. This follows from the fact that upper semi-continuity is equivalent to the closure of the hypograph, and the max function has a hypograph that is the union of the hypographs of the functions in the finite family.

Let $x\nleq y$, then there is $f\in \mathcal{F}$ such that $f(x)>f(y)$, which from formula (\ref{dis}) implies $d(x,y)=0$. From  formula (\ref{dis}) we have for every $x\in X$, $d(x,x)=0$.

Let $x\le y\le z$. As, by construction of $\le$, the functions are isotone, $d(x,y)=\textrm{inf}_{\mathcal{F}} [f(y)-f(x)]$ and similarly for the other pairs.
We have
\begin{align*}
\textrm{inf}_{\mathcal{F}}  [f(z)-f(y)]&+  \textrm{inf}_{\mathcal{F}} [f(y)-f(x)] \le \\
&\le \textrm{inf}_{\mathcal{F}} \{ [f(z)-f(y)]+[f(y)-f(x)]= \textrm{inf}_{\mathcal{F}} [f(z)-f(x)] ,
\end{align*}
which is the reverse triangle inequality.

Finally, let us prove $(iv)$. Since the functions in  $\mathcal{F}$  are isotone, the formula for $\tau$ is clear for $x\le y$. If $x \nleq y$ we can find $f\in \mathcal{F}$ such that $f(y)-f(x)<0$, but for every $\lambda>1$, $\lambda f\in \mathcal{F}$, thus $\inf_{\mathcal{F}} [f(y)-f(x)]=-\infty=\tau(x,y)$.

Coming to the expression for $\le^\downarrow$, we know that $(x,a)\le^\downarrow (y,b)$ is equivalent to $b\le a +\tau(x,y)$. So if  $(x,a)\le^\downarrow (y,b)$ we have for every $f\in \mathcal{F}$, using the just proved formula for $\tau$, $b\le a +\tau(x,y)\le a+f(y)-f(x)$. For the other direction if for every $f\in \mathcal{F}$,  $b\le a+f(y)-f(x)$, then taking the infimum over $\mathcal{F}$ gives $b\le a +\tau(y,x)$, which is the definition of $(x,a) \le^\downarrow (y,b)$.

Let $\mathcal{I}$ be the family of continuous isotone functions for $(X,\mathscr{T},\le)$. Since $\mathcal{F}\subset \mathcal{I}$, points (i) and (ii) show that $(X,\mathscr{T},\le)$ is completely regularly ordered, thus convex, hence locally convex. This proves that $(X,\mathscr{T},\le, d)$  is stable. The fact that the functions in $\mathcal{F}$ are continuous rushing functions follows immediately recalling the definition of topology, order and Lorentzian distance.

Since $\mathcal{F}\subset \mathcal{I}$, and the function in $\mathcal{I}$ are $\mathscr{T}$-continuous, the initial topology induced by $\mathcal{I}$ coincides with $\mathscr{T}$. Similarly, since  $\mathcal{F}\subset \mathcal{I}$ and the functions in $\mathcal{I}$ are isotone, the preorder induced by $\mathcal{I}$ coincides with $\le$. Thus $(X,\mathscr{T},\le)$ is completely regularly ordered.
\end{proof}

\begin{proposition}
 Denoting with $\mathcal{R}$ the continuous rushing functions of the  stable spacetime  $(X,\mathscr{T}_{\mathcal{F}},\le_{\mathcal{F}}, d_{\mathcal{F}})$ we have $\mathscr{T}_{\mathcal{F}}=\mathscr{T}_{\mathcal{R}}$, $\le_{\mathcal{F}}=\le_{\mathcal{R}}$ and  $d_{\mathcal{F}}=d_{\mathcal{R}}$.  Moreover, denoting with $\mathcal{I}$ the continuous isotone functions, we have $\mathscr{T}_{\mathcal{F}}=\mathscr{T}_{\mathcal{I}}$, $\le_{\mathcal{F}}=\le_{\mathcal{I}}$.
\end{proposition}
We shall see later (Thm.\ \ref{mai}) that we do not need to assume that the stable spacetime comes from a family of functions, provided suitable topological conditions are imposed.
\begin{proof}
Since the functions in $\mathcal{R}$ are $\mathscr{T}_{\mathcal{F}}$-continuous their initial topology is coarser than $\mathscr{T}_{\mathcal{F}}$, $\mathscr{T}_{\mathcal{R}}\subset \mathscr{T}_{\mathcal{F}}$, but since any function in ${\mathcal{F}}$ is continuous and rushing,  $\mathcal{R}\supset \mathcal{F}$, which implies $\mathscr{T}_{\mathcal{R}}\supset \mathscr{T}_{\mathcal{F}}$, thus $\mathscr{T}_{\mathcal{F}}=\mathscr{T}_{\mathcal{R}}$.

Since $\mathcal{R}\supset \mathcal{F}$ we have $\le_{\mathcal{R}}\subset \le_{\mathcal{F}}$. Conversely, if $x\le_{\mathcal{F}} y$, for any $f\in \mathcal{R}$, as it is $\le_{\mathcal{F}}$-isotone, $f(x) \le f(y)$, which implies $x\le_{\mathcal{R}} y$  and hence  $\le_{\mathcal{R}}\supset \le_{\mathcal{F}}$, thus $\le_{\mathcal{F}}=\le_{\mathcal{R}}$.

We already know that the orders induced by the families $\mathcal{F}$ and $\mathcal{R}$ coincide. Let us denote it $\le$. The Lorentzian distances induced by the families coincide outside $\le$, so we can assume $x\le y$. Since $\mathcal{R}\supset \mathcal{F}$ we have $d_{\mathcal{R}}(x,y)\le  d_{\mathcal{F}}(x,y)$, but every $f\in \mathcal{R}$ is $d_{\mathcal{F}}$-rushing (and isotone) which means $d_{\mathcal{F}}(x,y)\le f(y)-f(x)\ge 0$, and taking the infimum over $\mathcal{R}$, $d_{\mathcal{F}}(x,y)\le  d_{\mathcal{R}}(x,y)$, thus $d_{\mathcal{F}}=d_{\mathcal{R}}$.

The last statement was proved at the end of the proof of Thm.\ \ref{ckkr}.
\end{proof}

\begin{proposition} \label{mqpg}
On a spacetime $(X,\mathscr{T},\le, d)$ the $\mathscr{T}$-continuous, $\le$-isotone functions form a convex cone $\mathcal{I}$. The
$\mathscr{T}$-continuous, $\le$-isotone and $d$-rushing functions form a convex set $\mathcal{R}\subset \mathcal{I}$ such that $\lambda \mathcal{R}\subset \mathcal{R}$ for every $\lambda \ge 1$. Additionally, $\mathcal{I} + \mathcal{R} \subset \mathcal{R}$. For every $f,g\in \mathcal{R}$, $\textrm{min}(f,g), \textrm{max}(f,g)\in \mathcal{R}$, and a similar statement holds for $\mathcal{R}$ replaced by $\mathcal{I}$.
\end{proposition}

The set $\mathcal{I}$ is the analog of the causal cone on the tangent space of a Lorentzian manifold. The set $\mathcal{R}$ is the analog of the unit Lorentzian ball on the same tangent space. The nice fact is that the geometry of the tangent space seems reproduced at the functional level, where the functional space aims to reproduce the whole spacetime.

\begin{proof}
The set $\mathcal{I}$ is a cone because if $f$ is $\mathscr{T}$-continuous $\le$-isotone so is $\lambda f$ for every $\lambda \ge 0$. It is convex because if $f,g$ are $\mathscr{T}$-continuous $\le$-isotone then for $x\le y$ and any $t\in [0,1]$, as $f(x)\le f(y)$ and $g(x)\le g(y)$, we have  $(1-t)f(x)+tg(x)\le (1-t)f(y)+tg(y)$, namely $(1-t)f + t g$ is isotone.

Next, every rushing function satisfies for $x\le y$, $f(y) \ge f(x) + d(x,y) \ge f(x)$, so $\mathcal{R}\subset \mathcal{I}$. Moreover, if $f$ and $g$ are rushing from $f(y) \ge f(x) + d(x,y) $,   $g(y) \ge g(x) + d(x,y)$, multiplying the former equation by $(1-t)$, the latter equation by $t$ and summing we get that $(1-t)f + t g$ is rushing. Furthermore, multiplying the former equation by $\lambda\ge 1$,  $\lambda f(y) \ge \lambda f(x) + \lambda d(x,y)\ge  \lambda f(x) +  d(x,y)$, which proves that $\lambda f$ is rushing and so $\lambda \mathcal{R} \subset \mathcal{R}$. Let $f$ be rushing and $g$ be isotone, then for every $x,y\in X$, $x\le y$,  $f(y) \ge f(x) + d(x,y) $,   $g(y) \ge g(x)$ and summing we get that $f+g$ is a rushing function, hence $\mathcal{I} + \mathcal{R} \subset \mathcal{R}$.
Let $\epsilon = 1$ if $f, g \in \mathcal{R}$, and $\epsilon = 0$ if $f, g \in \mathcal{I}$. Then, for $x \le y$, we have:
%If $f$ and $g$ are rushing, let us set $\epsilon =1$, while if they are isotone, let us set $\epsilon =0$. Note that for $x\le y$,
\begin{align*}
f(y)&\ge f(x)+\epsilon d(x,y)\ge  \textrm{min}(f,g)(x)+ \epsilon d(x,y), \\
g(y)&\ge g(x)+\epsilon d(x,y)\ge  \textrm{min}(f,g)(x)+ \epsilon d(x,y),
\end{align*}
thus  $\textrm{min}(f,g)(y) \ge \textrm{min}(f,g)(x)+ \epsilon d(x,y)$ hence  $\textrm{min}(f,g)$ is rushing/isotone. Similarly, for
$x\le y$,
\begin{align*}
\textrm{max}(f,g)(y) \ge f(y)&\ge f(x)+\epsilon d(x,y), \\
\textrm{max}(f,g)(y) \ge g(y)&\ge g(x)+\epsilon d(x,y),
\end{align*}
thus  $\textrm{max}(f,g)(y) \ge \textrm{max}(f,g)(x)+ \epsilon d(x,y)$ hence  $\textrm{max}(f,g)$ is rushing/isotone.
\end{proof}

\begin{remark}
We recall that a cone is sharp if it does not contain any line passing through its vertex (origin of the vector space).
The cone $\mathcal{I}$ is not necessarily sharp. If a  not identically zero function $f \in \mathcal{I}$ belongs to a line passing through the origin and contained in $\mathcal{I}$, then $-f$ belongs to $\mathcal{I}$, too. But they would be both isotone, which would imply for any pair $x\le y$, $f(x)=f(y)$.
Under a mild connectedness assumption (e.g., that the preorder is connected in the sense that any two points are connected by a chain of pairs of distinct  preorder-comparable points), then $f$ must be constant.
%If any two distinct points are connected by a chain of pairs of distinct but preorder-comparable points (a mild connectedness assumption), then $f$ is a constant.
Later we shall introduce the time functions which form a cone (not including  the origin) $\mathcal{T}\subset \mathcal{I}$ which is sharp whenever the order is not the discrete order.
%If the functions in $\mathcal{I}$ represent the order, in the sense that Eq.\ (\ref{ord}) holds, then this would be possible only for the discrete order.
\end{remark}

\section{Recovering spacetime from the continuous rushing functions} \label{cnxp}

The objective of this section is to prove a converse of the previous result. Not only a family of functions induces a stable spacetime, where the elements of the family are a posteriori interpreted as continuous rushing functions. The converse is true: under weak topological assumptions any stable spacetime can be recovered from the continuous rushing functions over it.

The proof passes through the formulation of spacetime via the product trick.

For generality we shall write the following results in terms of the $k_\omega$-space condition.
%It is weaker than ``local compactness $\sigma$-compactness'', and possibly better behaved, but less known.
A $k_\omega$-space can be characterized through the
following property: there is a countable sequence $K_i$ (called {\em admissible}) of compact sets such that $\bigcup_{i=1}^\infty K_i=X$ and for every $O\subset X$, $O$ is open if and only if $O\cap K_i$ is open in $K_i$ in the induced topology. In this definition we do not assume $X$ to be Hausdorff but this fact will not be relevant since the spaces to which we shall apply the condition will be closed ordered topological spaces hence Hausdorff.

Without loss of generality the admissible sequence can be assumed increasing, indeed if $K_i$ is admissible then $\tilde K_n=\bigcup_{i=1}^n K_i$ is admissible. In this case every compact set $C$ is contained in some element of the admissible increasing sequence \cite[Lemma 9.3]{steenrod67}.

We recall that \cite{franklin77,minguzzi11c}
\begin{center}
compact $\Rightarrow$   $k_\omega$-space $\Rightarrow$ $\sigma$-compact $\Rightarrow$   Lindel\"of
\end{center}
and the fact that local compactness makes the last three properties coincide (for us a topological space is locally compact if every point admits a compact neighborhood).
In particular, locally compact $\sigma$-compact spaces (e.g.\ topological manifolds) are $k_\omega$-spaces.
%The quotient of a $k_\omega$-space is a $k_\omega$-space, and every $k_\omega$-space is the quotient of a locally compact $\sigma$-compact space (Morita's theorem) \cite{morita56,minguzzi11f}.
Every Hausdorff $k_\omega$-space is paracompact and normal \cite{morita56}.  The quotient of a $k_\omega$-space is a $k_\omega$-space \cite{morita56} \cite[Thm.\ 1.1]{minguzzi11f}.
 The product of two $k_\omega$-spaces $X, Y$ is a $k_\omega$-space. If $X_i$ and $Y_i$ are admissible sequences for $X$ and $Y$, respectively, then $X_i\times Y_i$ is an admissible sequence for $X\times Y$ \cite{franklin77}.
An important property that will be used in the following proof is \cite[Thm.\ 2.7]{steenrod67}: A function $f: X \to \mathbb{R}$ is continuous if for each $i$, $f\vert_{K_i}$ is continuous, where each $K_i$ is assigned the induced topology.

\begin{lemma} \label{cjqr}
The spacetime  $(X,\mathscr{T},\le,d)$ is locally compact ($\sigma$-compact, $k$-preserving, a $k_\omega$-space) iff so is $(X^\times,\mathscr{T}^\times,\le^\downarrow)$.
\end{lemma}

\begin{proof}
The product of two locally compact ($\sigma$-compact) spaces is locally compact (resp.\ $\sigma$-compact), thus if $X$ is locally compact ($\sigma$-compact) then so is $X^\times$. Since the continuous image of a compact set is compact, if $X^\times$ is locally compact ($\sigma$-compact) then so is $X$.

%As $X^\times= X\times \mathbb{R}$ is endowed with the product topology $\mathscr{T}^\times$, then $X^\times$ is clearly $\sigma$-compact as $X^\times=\cup_{i\in \mathbb{Z}} \{ X\times [i,i+1]\}$. It is also locally compact, as every point in $X^\times$ belongs to an open set of the form $X\times (a,b)$ which is contained in the compact set $X\times [a,b]$.

If $X$ is $k$-preserving then $X^\times$ is $k$-preserving, indeed in every closed ordered space, for any compact set $C$ we have that $i(C)\cap d(C)$ is closed \cite[Prop.\ 4, p.\ 44]{nachbin65}, thus we need only to show that $i^\downarrow (K) \cap d^\downarrow(K)$ is contained in a compact set, where $K\subset X^\times$ is any compact set. But $\pi_2(K)$ is compact hence bounded, thus there is $M>0$ such that it is contained in $[-M,M]$.  The set $i(\pi_1(K))\cap d(\pi_1(K))$ is compact because $X$ is $k$-preserving. Let $D<+\infty$ be the maximum of $d$ on it, then $\pi_2(i^\downarrow (K) \cap d^\downarrow(K))=[-M-D,M+D]$ thus
\[
i^\downarrow (K) \cap d^\downarrow(K) \subset [i(\pi_1(K))\cap d(\pi(K))] \times [-M-D,M+D]
\]
where the set on the right-hand side is a product of compact sets hence compact.

The other direction follows noting that if $K$ is a compact set on $X$, $i(K)\cap d(K)$ is  closed and
\[
i(K)\cap d(K) \subset \pi_1\Big(i^\downarrow(K\times\{0\})\cap d^\downarrow(K\times\{0\}) \Big),
\]
where the right-hand side is compact.
The statements on the $k_\omega$-property follow from its preservation under quotient and Cartesian product.
\end{proof}

\begin{lemma} \label{lem}
Let $(X,\mathscr{T},\le,d)$ be a weakly stable $k_\omega$-spacetime.
%Let $(X,\mathscr{T},\le,d)$ be a locally compact and $\sigma$-compact stable spacetime.
%second-countable and
%The   closed ordered space $(X^\times,\mathscr{T}^\times,\le^\downarrow)$ is locally compact and $\sigma$-compact (moreover, $X$ is $k$-preserving iff so is $X^\times$).
 There is a continuous isotone function $F: X^\times \to \mathbb{R}$ such that
$F((x,r))=h(x)-r$,
 %$\pi_1({-1}(0))=X$, $H^{-1}(0)=\{(x,h(x)), x\in X\}$
 where $h:X\to \mathbb{R}$ is a continuous rushing function. Conversely, for every continuous rushing function $h:X\to \mathbb{R}$, the function $F: X^\times \to \mathbb{R}$ defined as above is continuous and isotone.
%for each $x\in X$, $H((x,r))\to \pm \infty$ for $r\to \pm \infty$.
%Finally, if $\sup_{X\times X} d< B$ for some constant $B>0$, then $h$ can be found so that  for every $x,y\in X$, $\vert h(y)-h(x)\vert<B$ (thus bounded).
\end{lemma}

%In part the problem is to prove that there is an isotone function on $X^\times$ which is strictly decreasing over the fibers (in the real line order). This is the case for utilities, whose existence is  assured under second-countability \cite{minguzzi11f} . However, we do not need to make such assumption.
% a fact which complicates the proof.

The isotone functions on $X^\times$ of the form $F((x,r))=h(x)-r$ will be said to be {\em translationally invariant}.

\begin{proof}
As $X$ is a $k_\omega$-space there is an admissible sequence of compact sets $K_i$ such that $K_i\subset K_{i+1}$,  $\cup_i K_i=X$.
% and $K_i\subset \textrm{Int} K_{i+1}$.

Let $D_i$ be constants larger than the maximum of $d\vert_{K_i\times K_i}$. No point in the future of  $K_i \times \{0\}$ can comprise a point in  $K_i \times \{D_i\}$, which means that the future closed set $i^\downarrow( K_i \times \{r\})$ does not intersect the past closed set  $d^\downarrow( K_i \times \{r+D_i\})$ where $r$ is any constant. The arbitrariness of $r$ follows from translational invariance of the order $\le^\downarrow$.
For each  index value $i$ we make a choice for $r$, $r_i<0$, so that $r_i+D_i$ is negative and decreasing and $r_i, r_i+D_i \to -\infty$.

Every closed preordered $k_\omega$-space is normally preordered \cite[Thm.\ 2.7]{minguzzi11f}. By Nachbin's generalization of the Urysohn lemma, there is a continuous isotone function $\hat G_i: X^\times \to [0,1]$ such that $ \hat G_i^{-1}(0)\supset d^\downarrow( K_i \times \{r_i+D_i\})\supset  K_i \times \{r_i+D_i\}$, $\hat G^{-1}(1)\supset i^\downarrow( K_i \times \{r_i\})\supset K_i \times \{r_i\}$.

Similarly, for the same index value $i$,  the future closed set $i^\downarrow( K_i \times \{-r_i-D_i\})$ does not intersect the past closed set  $d^\downarrow( K_i \times \{-r_i\})$, thus there is a continuous isotone function $\check G_i: X^\times \to [-1,0]$ such that $ \check G_i^{-1}(-1)\supset d^\downarrow( K_i \times \{-r_i\})\supset  K_i \times \{-r_i\}$, $\hat G_i^{-1}(0)\supset i^\downarrow( K_i \times \{-r_i-D_i\})\supset K_i \times \{-r_i-D_i\}$.

The function $G_i=\check G_i+\hat G_i: X^\times \to [-1,1]$ has the following properties $G_i^{-1}(1)\supset K_i \times (-\infty, r_i]$, $G_i^{-1}(0)\supset K_i \times [r_i+D_i, -r_i-D_i]$,  $G_i^{-1}(-1)\supset K_i \times [-r_i, +\infty )$.  Since $\mathbb{R}$ is a $k_\omega$-space, the topological space $X^\times$ is a $k_\omega$-space and an admissible sequence of compact sets is provided by $C_i:=K_i \times [r_i+D_i, -r_i-D_i]$.
The function $G=\sum_i G_i$ is well defined because at every point only a finite number of terms  are non-zero. Moreover, it is continuous on $C_n$ because it is there equal to $\sum_{i=1}^{n-1} G_i\vert_{C_i}$ which is continuous, being the terms finite in number. By the mentioned property of $k_\omega$-spaces, $G$ is continuous. The isotone property is clear.
Moreover, for each $x$, $G((x,y))\to\pm \infty$ for $y \to \mp \infty$.
%This is the desired asymptotic behaviour, however, it is not decreasin on the fibers.  We now construct a function $U$ which this last property.

%
%
%Let us consider the compact sets $d^\downarrow( K_i \times \{n_i\})$, $d^\downarrow( K_i \times \{m_i\})$
%
%
%A locally compact $\sigma$-compact closed ordered space is normally ordered \cite{minguzzi11f}. By the bound on $d$, denoting $X_a=X\times \{a\}$, we have that $i(X_a)\cap X_{a+N}=\emptyset$, thus by Nachbin's generalization of the Urysohn lemma, there is a continuous isotone function $G$ such that $ G^{-1}(-1)\supset d(X_{a+N})\supset X_{a+N}$, $G^{-1}(1)\supset i(X_a)\supset X_a$.

This paragraph is devoted to the proof that there is a continuous isotone function $U: X^\times \to [-1,1]$ that is decreasing on the fibers (in their real line order).\footnote{If $X$ is second-countable this step can be simplified, as the product $X^\times$ is second-countable too. By \cite[Thm.\ 5.5]{minguzzi11f} there is a continuous utility $U: X^\times \to [-1,1]$, which hence  decreases over the fibers.}
Let $n\in \mathbb{N}$  and $x\in X$, as $d^\downarrow((x,0))\cap i^\downarrow((x,-1/2^n))=\emptyset$ (because $d(x,x)=0$) we can find, again by order normality, a continuous isotone function $F^n_x: X^\times \to [-1,1] $ such that $(F^n_x){}^{-1}(-1)\supset  d^\downarrow((x,0)) \ni (x,0)$,  $(F^n_x){}^{-1}(1)\supset  i^\downarrow((x,-1/2^n)) \ni (x,-1/2^n)$. By continuity, there is an open neighborhood $O_x\ni x$ such that $F^n_x((y,0))<F^n_x((y,-1/2^n))$, for every $y\in O_x$. As $X$ is a $k_\omega$-space it is Lindel\"of, so  from $\{O_x: x\in X\}$ we can extract a countable covering $\{O_{x_i}, i\in \mathbb{N}\}$, and then the function $F^n:X^\times \to [-1,1]$
\[
F^n=\frac{1}{2^i} \sum_{i=1} F^n_{x_i},
\]
by uniform convergence, is a continuous isotone function such that for every $x\in X$, $F^n((x,0))<F^n((x,-1/2^n))$. Note that if $h((x,y))$ is a continuous isotone function, so is, for every $q\in \mathbb{R}$, $h'((x,y))=h((x,y-q))$,  by translational invariance of $\le^\downarrow$.
Let $q: \mathbb{N}\to \mathbb{Q}$, $m \mapsto q_m$, be a bijection, and consider the function
\[
U((x,r))=\sum_{m=1 ,n=1} \frac{1}{2^{m+n}} F^n((x,r-q_m)).
\]
This series converges uniformly to a continuous and isotone function $U:X^\times \to [-1,1]$.

Note that by continuity of $F^n$, the inequality  $F^n((x,0))<F^n((x,-1/2^n))$ implies that for each $x$ there is some $U_x\subset \mathbb{R}$, neighborhood of $0$, such that for every $s\in U_x$, $F^n((x,s))<F^n((x,s-1/2^n))$.

Let $x\in X$, $r_1,r_2\in \mathbb{R}$, $r_1<r_2$. We can find a rational number $q_m$, and $n\in \mathbb{N}$ such that $r_2-q_m\in U_x$, $r_1-q_m<r_2-q_m-1/2^n <r_2-q_m$
 %$[q_m-\frac{1}{2^n}, q_m ] \subset (r_1,r_2)\cap U_x$,
 from which it follows that
\[
F^n((x,r_1-q_m))  \ge F^n((x,r_2-q_m-\frac{1}{2^n}))>F^n((x,r_2-q_m))
%F^n((x,r_1-q_m))<F^n((x,r_2-q_m)),
\]
which implies $U((x,r_1))> U((x,r_2))$, thus, $U$ is decreasing in the fibers (in their real line order).

%As $X$ is second-countable, the product $X^\times$ is second-countable too. By \cite[Thm.\ 5.5]{minguzzi11f} there is a continuous utility $U: X^\times \to [-1,1]$. Thus the function $H:=G+\frac{1}{2} U$ is a continuous utility (thus it is  decreasing on the fibers) which satisfies $H^{-1}(0) \subset X \times [-N/2, N/2]$ and $\pi_1(H^{-1}(0))=X$ (just choose $a=-N/2$).

The function $H:=G+\frac{1}{2} U$ is a continuous isotone function that has the same asymptotic behavior of $G$ on the fibers, and is decreasing over them (with their canonical order), which implies that for every $x\in X$ there is just one point $(x,h(x))$ for which $H$ vanishes.
%  decreases on the fibers and  satisfies $H^{-1}(0) \subset X \times [-N/2, N/2]$ and $\pi_1(H^{-1}(0))=X$ (just choose $a=-N/2$). As a
As a consequence, there is a function $h: X\to \mathbb{R}$ such that $H^{-1}(0)=\{(x,h(x)): x\in X\}$. Observe that by the continuity of $H$, $H^{-1}([0,+\infty))=\{(x,r): x\in X, r\le h(x)\}$, is closed which is the hypograph of $h$. As a consequence, $h$ is upper semi-continuous.
Similarly,  by the continuity of $H$,  $H^{-1}((-\infty,0])=\{(x,r): x\in X, r\ge h(x)\}$ is closed, which is the epigraph of $h$, thus $h$ is lower semi-continuous and hence continuous.

%If it were not lower semi-continuous, there would be $x\in X$, and $\epsilon >0$ such that for every open neighborhood $O\ni x$, we could find $y\in O$ with $h(y)<h(x)-\epsilon$, which means $(y, h(x)-\epsilon) \in H^{-1}((-\infty, 0])$, which by closure of the latter set implies $(x,h(x)-\epsilon)\in  H^{-1}((-\infty, 0])$.
%This contradicts the fact that $H$ decreases over the fibers.

%We conclude that $h$ is continuous.
Finally, let us consider two points $x,y\in X$, $x\le y$, then $(x,a) \le^\downarrow (y, a+ d(x,y) )$. Choosing $a$ so that $(x,a)\in H^{-1}(0)$, that is $a=h(x)$, we must have $H((y,a+ d(x,y))) \ge 0$, that is $a+d(x,y)\le h(y)$, which proves that $h$ is a rushing function.

The function $F((x,r)):=h(x)-r$ is now continuous and isotone because, $(x,a)\le^\downarrow (y,b)$ implies $x\le y$ and $b\le a+d(x,y)$, which implies
\[
F((x,a))=h(x)-a\le h(x)+d(x,y)-b\le h(y)-b=F((y,b)).
\]
%It remains to prove that if $\sup_{X\times X} d< B$ then $h$ can be found bounded. This is really a simplified version of the previous arguments. The construction of function $G$ is now simplified, we do not need it to have image on the whole real line on every fiber. Consider the subsets of $X^\times$, $X_0=X\times \{0\}$ and $X_B=X\times \{B\}$, then $i^\downarrow(X_0)$.
%Let us denote the generic element of $X^\times$ with $(x,r)$. We know that $(x,N) \nleq (y,0)$ for any $x,y\in X$, thus there is some continuous isotone function $f_{xy}$ such that $f_{xy}(x)=1$, $f_{xy}(y)<0$. By continuity there is a whole open neighborhood $U(y)$ of $y$ such that if $y'\in U(y)$, $f_{xy}(y')<0$. We can cover $X_0:=X\times \{0\}$ with open neighborhoods $\{U(y_i), i=1,\cdots,k\}$, so that $\ma
\end{proof}
%\subsection{Example subsection heading}

Let us denote $x\sim y$ if $x\le y$ and $y\le x$. Let us also denote $x<y$  if $x\le y$ and $y\nleq x$.
A  isotone function $f: X\to \mathbb{R}$ such that $x<y \Rightarrow f(x)<f(y)$ is a {\em utility}.  If $\le$ is a partial order then a utility is defined by ``$x\le y$ and $x\ne y \Rightarrow f(x)<f(y)$''. If there is a function with the last property then $\le$ is necessarily a partial order. We call a continuous function that satisfies  ``$x\le y$ and $x\ne y \Rightarrow  f(x)<f(y)$'' {\em time function} (we avoid {\em strictly isotone} since some authors denote with this term a utility).

Our terminology is consistent with the result for the smooth case, since when $K$ is antisymmetric (stable causality) its continuous utilities are precisely the time functions \cite[Thm.\ 6]{minguzzi09c}, and if there is a time function $K$ is antisymmetric.

\begin{theorem} \label{cmgp}
A second-countable   closed preordered $k_\omega$-space (i.e.\ causal structure) $(X,\mathscr{T},\le)$ is weakly stably causal (i.e.\ an order) iff it admits a time function.
%In this case, let $\mathcal{F}\ne \emptyset$ be the family of time functions. Then $\mathscr{T}=\mathscr{T}_{\mathcal{F}}$ and $\le=\le_{\mathcal{F}}$.
\end{theorem}

It is useful to observe that every second-countable locally compact space is a second-countable $k_\omega$-space. The converse holds if $(X,\mathscr{T},\le)$ is locally convex or the topology is Hausdorff \cite[Prop.\ 3.1]{minguzzi11c}.

\begin{proof}
%To the right.
%Since the relation is closed and the order is antisymmetric the topology is Hausdorff (but this is not really needed for the proof \cite{minguzzi11c}).
%A second-countable locally compact topological space is $\sigma$-compact.
By the generalization of Auslander-Levin's theorem \cite{levin83,bridges95}  proved in \cite[Cor.\ 1.4]{minguzzi11c}
every second-countable closed preordered $k_\omega$-space admits a utility function, thus every second-countable closed ordered $k_\omega$-space admits a time function. Conversely, if there is a time function $\le$ is an order.
\end{proof}

\begin{proposition}
A   weakly stably causal $k_\omega$-spacetime $(X,\mathscr{T},\le,d)$ is weakly stable iff it admits a continuous rushing function.
\end{proposition}

\begin{proof}
The implication to the right follows from Lemma \ref{lem}. For the converse, let $f$ be a continuous rushing function then, for every $x,y\in X$, $x\le y$, $f(x)+d(x,y)\le f(y)<+\infty$, which proves that $d(x,y)$ is finite.
\end{proof}

We say that $f:X\to \mathbb{R}$ is a rushing time function if it is both a rushing function and a time function.
%and  ``$x\le y$ and $x\ne y\Rightarrow f(x)+d(x,y)<f(y)$.

The next result is the metric analog of the first statement (non-strict version) in \cite[Thm.\ 4.6]{minguzzi17}.
\begin{theorem}
A second-countable $k_\omega$-spacetime $(X,\mathscr{T},\le,d)$ is weakly stable iff it admits a rushing time function.
\end{theorem}

\begin{proof}
$\Rightarrow$.
%Second-countability and local compactness implies $\sigma$-compactness.
By Lemma \ref{lem} there is a continuous rushing function $h$. By \cite[Cor.\ 1.4]{minguzzi11c} there is a continuous utility $g$, thus $f:=h+g$ is a rushing time function.

$\Leftarrow$. The existence of a time function implies that $\le$ is an order.  Since $f$ is rushing, from $x\le x$, $f(x)+d(x,x)\le f(x)$, which implies $d(x,x)=0$. Thus $(X,\mathscr{T},\le,d)$  is weakly stably causal. Finally,  if $x \le y$, $f(x)+d(x,y)\le f(y)$ shows that $d(x,y)$ is finite, which proves that $(X,\mathscr{T},\le,d)$ is weakly stable.
\end{proof}

%
%The next result is \cite[Thm.\ 3.3]{minguzzi12d}
%\begin{theorem} \label{box}
%Let $(X,\mathscr{T},\le)$ be a locally compact $\sigma$-compact  $k$-preserving closed ordered space, then $(X,\mathscr{T},\le)$
%%is normally ordered and hence the order is represented by the set of continuous isotone functions $\mathcal{I}$: $\le=\le_{\mathcal{I}}$. Moreover, assume that additionally the order is $k$-preserving, then  the topology is convex and hence $(X,\mathscr{T},\le)$
%is completely regularly ordered: $\mathscr{T}=\mathscr{T}_{\mathcal{I}}$ and $\le=\le_{\mathcal{I}}$, where $\mathcal{I}$ is the family of continuous isotone functions.
%\end{theorem}

The next result is \cite[Cor.\ 2.14]{minguzzi12d}

\begin{theorem}\label{box2}
Let $(X,\mathscr{T},\le)$ be a locally compact $\sigma$-compact locally convex closed preordered space, then $(X,\mathscr{T},\le)$ is completely regularly preordered: $\mathscr{T}=\mathscr{T}_{\mathcal{I}}$ and $\le=\le_{\mathcal{I}}$, where $\mathcal{I}$ is the family of continuous isotone functions.
\end{theorem}

\begin{lemma} \label{cnqx}
Let $(X,\mathscr{T},\le,d)$ be a locally compact  $\sigma$-compact stable spacetime. Let $\mathcal{T}$ be the family of translationally invariant continuous isotone functions on $(X^\times, \mathscr{T}^\times, \le^\downarrow)$. Then $\mathscr{T}^\times= \mathscr{T}_{\mathcal{T}}$, $\le^\downarrow=\le_{\mathcal{T}}$.
\end{lemma}

\begin{proof}
By Lemma \ref{cjqr} $X^\times$ is locally compact and $\sigma$-compact. By Prop.\ \ref{otyu}  $(X^\times, \mathscr{T}^\times, \le^\downarrow)$ is locally convex. By Thm.\ \ref{box2} we have that the topology and order on $X^\times$ are generated by the continuous isotone functions. Now, the key observation is that these properties  do not depend on the  functions themselves but rather on their level sets, and that for every level set which is a graph we can find a translationally invariant continuous isotone function with the same level set.

Let us substatiate this idea. Let $(x,p) \nleq^\downarrow (y,q)$, we know that there is a continuous isotone function $R$ such that $R((x,p))> R((y,q))$.
%Without loss of generality we can assume $R((x,p))=1$, $R((y,q))=-1$ (otherwise redefine $R\to a R+b$, for suitable $a>0$ and $b$).
By Lemma \ref{lem}  there is a translationally invariant continuous isotone function $F$ (translational invariance implies that it decreases over the fibers), thus the continuous isotone function $H:=a [\arctan R+b F]+c$, for suitable constants $a,b>0$ and $c$,  is such that $H((x,p))=1$, $H((y,q))=-1$. Indeed, take $b$ sufficiently small so that  $(\arctan R+b F)((x,p)) > (\arctan R+b F)((y,q))$ and then solve for $a$ and $c$.
Observe that $H$ is  decreasing on the fibers and that on every fiber it has image the whole real line, thus the zero-level set is indeed a graph: $H^{-1}(0)=\{(x,h(x)): x\in X\}$, $h:X \to \mathbb{R}$. The proof that $h$ is a continuous rushing function goes as in Lemma \ref{lem}. Now, $H'((x,r)):=h(x)-r$ has the same zero-level set, and the same regions in which it is positive and negative of $H$. We conclude that $H'((x,p))> H'((y,q))$ and thus $\le^\downarrow=\le_{\mathcal{T}}$.

As for the representation of the topology, let us consider a point on $(x,\bar r)\in X^\times$, and a product open neighborhood $O^\times:=O\times (\bar r-\epsilon, \bar r+\epsilon)$, $x\in O$. Without loss of generality we can assume that $O$, and hence $O^\times$, is relatively compact.

We know that there are continuous isotone bounded functions $G, H:X^\times \to \mathbb{R}$, such that $ H^{-1}((0, +\infty))\cap G^{-1}((-\infty, 0))\subset O^\times$, $H((x,\bar r))=-G((x,\bar r))=1$. The problem is that the zero level sets of $G$ and $H$ might not be graphs over $X$, on the one hand because they do not project on the whole $X$, on the other hand because they intersect the same fiber in more than one point. Much of the remaining proof will be devoted to modify the functions so as to solve this problem.

Let $F$ be a translationally invariant function such that $F((x,\bar r))=0$, and let $u(t):=e^t-1$, then $G'=2G+ u(F)+1$, has the same properties of $G$, as $G'((x, \bar r))=-1$ and $G'^{-1}((-\infty, 0))\subset G^{-1}((-\infty, 0))$, but additionally, it is decreasing (in the order of $\mathbb{R}$) over every fiber and it converges to $+\infty$ as the real coordinate goes to $ -\infty$ (in the other direction it stays bounded).
%By continuity and compactness, for each $i$ there is some $g_i\in\mathbb{R}$, $g_i<\bar r$, such that $G'$ is  positive on $K_i\times (-\infty,g_i]$.

Similarly, there is a continuous isotone function $H'=2H-u(-F)-1$  which goes to $-\infty$ over every fiber when the real parameter goes to $+\infty$, is decreasing over every fiber, $H'((x, \bar r))=1$, and  $H'{}^{-1}((0, +\infty))\subset H^{-1}((0, +\infty))$ (in the other direction it stays bounded).
%By continuity and compactness, for each $i$ there is some $h_i\in\mathbb{R}$, $h_i>\bar r$, such that $H'$ is  negative on $K_i\times [h_i,+\infty)$.

The zero level sets of the functions $H'$, $G'$ are graphs over their projection on $X$, but these projected subsets might differ from $X$.

%The problem is to keep the same properties while getting a compact zero-level set. Observe that we obtain a similar function $H'$ which goes to $-\infty$ over every fiber (in the order of $\mathbb{R}$), and a constant $h\in \mathbb{R}$, $h>\bar r$, such that $H'$ is  negative on $X_h$ and  on the region $r\ge h$.

Let $M<+\infty$ be larger than the maximum of  $-H'$ and $G'$ on $\overline{O^\times}$. The function
\[
G''=G'+\min\{H'+M,0\}
\]
over every fiber converges to $-\infty$ when the real parameter of the fiber goes to $+\infty$, because both terms are  non-increasing along the fiber (with the order induced from $\mathbb{R}$) and $H'\to -\infty$. When instead the parameter goes to $-\infty$ the function  $G'$ goes to $+\infty$, but  $H'$ increases,  $\min\{H'+M,0\}$ remains bounded  so $G''$ goes to $+\infty$.  We have $G''^{-1}((-\infty, 0))\supset G'^{-1}((-\infty, 0))$ but there is change from $G'$ to $G''$ only where $H'<-M<0$, thus
\[
H'^{-1}((0, +\infty))\cap G''^{-1}((-\infty, 0))=  H'^{-1}((0, +\infty))\cap G'^{-1}((-\infty, 0)).
\]
Since $G'$ and $G''$ coincide on $\overline{O^\times}$ the constant $M$ is also an upper bound for $G''$ on that set.
The function
\[
H''=H'+\max\{G''-M,0\}.
\]
over every fiber converges to $+\infty$ when the real parameter of the fiber goes to $-\infty$, because both contributions are non-decreasing and $G''\to +\infty$. When instead the parameter goes to $+\infty$ the functions  $H', G''$ go to $-\infty$,
%but  $G''$ decreases,  $\max\{G''-M,0\}$ remains bounded
so $H''$ goes to $-\infty$.  We have $H''^{-1}((0, +\infty))\supset H'^{-1}((0, +\infty))$ but there is change from $H'$  to $H''$  only  where $G''>M>0$, thus
\begin{align*}
&H''^{-1}((0, +\infty))\cap G''^{-1}((-\infty, 0))=  H'^{-1}((0, +\infty))\cap G''^{-1}((-\infty, 0))  \\
=& \, H'^{-1}((0, +\infty))\cap G'^{-1}((-\infty, 0)) \subset H^{-1}((0, +\infty))\cap G^{-1}((-\infty, 0))\subset O^\times.
\end{align*}
The zero level sets $G''^{-1}(0)$, $H''^{-1}(0)$ are graphs over $X$.
As a last step we replace $G''$ and $H''$ with translationally invariant functions sharing the same zero-level sets (see the last step in the proof of Lemma \ref{lem}).
\end{proof}

We are ready to prove the main result of this work

\begin{theorem} \label{mai}
On a
locally compact $\sigma$-compact stable spacetime $(X,\mathscr{T},\le,d)$, denoting with $\mathcal{R}$ the family of continuous rushing functions, we have $\mathcal{R}\ne \emptyset$, $\mathscr{T}=\mathscr{T}_{\mathcal{R}}$,
$\le=\le_{\mathcal{R}}$ and $d=d_{\mathcal{R}}$. Under second-countability we can replace the set of continuous rushing functions with the set of rushing time functions.
\end{theorem}

\begin{proof}
The inequality $\mathcal{R}\ne \emptyset$ follows from Lemma \ref{lem}.

Note that under second-countability by weak stable causality there is a time function $u: X \to [-1,1]$, see Thm.\ \ref{cmgp}. This fact will be used in what follows to prove the last statement.

Suppose that $x\nleq y$, then $(x,0)\nleq^\downarrow (y,0)$. By the equality $\le^\downarrow=\le_{\mathcal{T}}$ proved in Lemma \ref{cnqx}, there is a translationally invariant continuous isotone function $F$ such that $F((x,0))> F((y,0))$. As $F((z,r))=f(z)-r$, with $f$ continuous rushing function, we have $f(x)>f(y)$, which proves $\le=\le_{\mathcal{R}}$. Under second-countability replace $f$ with $f'=f+\epsilon u$, $\epsilon>0$, which is a rushing time function, then for small $\epsilon$, we have still $f'(x)>f'(y)$.

Let $x\le y$ and let $a>d(x,y)$, then $(x,0) \nleq^\downarrow (y,a)$ which, by $\le^\downarrow=\le_{\mathcal{T}}$, implies that there is a translationally invariant continuous isotone function $F$ such that $F((x,0))> F((y,a))$. As $F((z,r))=f(z)-r$, with $f$ continuous rushing function, we have $f(x)>f(y)-a$, that is $a>f(y)-f(x)\ge d(x,y)$, which proves, by the arbitrariness of $a$,
\[
d(x,y)= \inf_{\mathcal{R}}[f(y)-f(x)]=\max\{0,\inf_{\mathcal{R}}[f(y)-f(x)]\} =d_{\mathcal{R}}.
\]
For $x\nleq y$ we have necessarily $(x,0) \nleq^\downarrow (y,0)$ which, by $\le^\downarrow=\le_{\mathcal{T}}$, implies that there is a translationally invariant continuous isotone function $F$ such that $F((x,0))> F((y,0))$. As $F((z,r))=f(z)-r$, with $f$ continuous rushing function, we have $f(x)>f(y)$, thus $\max\{0,\inf_{\mathcal{R}}[f(y)-f(x)]\} =0$, which proves that the formula holds in all cases.  Under second-countability, prove the first part by replacing $f$ with $f'=f+\epsilon u$ with $\epsilon>0$ so small that $a>f'(y)-f'(x)\ge d(x,y)$ still holds. Similarly, in the second part choose it so that we still obtain $f'(x)>f'(y)$.

Let $x\in X$ and let $O$ be an open relatively compact neighborhood of $x$. Let $\epsilon>0$, and let $O^\times=O\times (-\epsilon,\epsilon)$. By Lemma  \ref{cnqx} we can find $F,G\in \mathcal{T}$ such that $(x,0)\in \{F>0\}\cap \{G<0\}\subset O^\times$. But $F((z,r))=f(z)-r$, $G((z,r))=g(z)-r$, thus  $g(x)<0<f(x)$, and we have the implication: for all $z\in X$, $r\in \mathbb{R}$, $g(z)<r<f(z) \Rightarrow \vert r \vert <\epsilon$ and $z\in O$. This last property implies $\{f-g>0\} \subset O$ and  $\sup_{z: f(z)-g(z)>0} \max \{f(z),-g(z)\} \le \epsilon$ which implies $\sup\{f-g\}\le 2 \epsilon$.
% $x\in \{f(z)-g(z)>0\} \subset O$, (also observe that $ f-g<2\epsilon$)
The inclusion proves $ \mathscr{T}_{\mathcal{R}-\mathcal{R}}\supset \mathscr{T}$. The inclusion $\mathscr{T}_{\mathcal{R}}\supset \mathscr{T}_{\mathcal{R}-\mathcal{R}}$ follows from the fact that the functions of the form $f-g$, $f,g\in \mathcal{R}$ are $\mathscr{T}_{\mathcal{R}}$-continuous. The inclusion $\mathscr{T}\supset \mathscr{T}_{\mathcal{R}}$ is due to the $\mathscr{T}$-continuity of  the functions in $\mathcal{R}$. We conclude, $\mathscr{T}_{\mathcal{R}}=\mathscr{T}$.
Under second-countability, replace $f$ with $f'=f+\epsilon u-\epsilon$ and $g$ with $g'=g+\epsilon u+\epsilon$ with $\epsilon$ so small that we have still  $g'(x)<0<f'(x)$. But $f'-g'=f-g-2\epsilon$, thus $x\in \{f'-g'>0\}\subset  \{f-g>0\}\subset O$. The rest of the argument does not change.
\end{proof}

%\begin{remark}
%As for the topology, as the above proof shows, we have actually a stronger result. For every open neighborhood $O$ of $x\in X$, we can find $f,g\in \mathcal{R}$ such that $x\in \{f-g>0\}\subset O $, and moreover, for any $\epsilon>0$ we can find the functions so that $f-g<\epsilon$.
%\end{remark}

The theorem applies also to   closed ordered spaces, just set  $d=0$ so that $\mathcal{R}=\mathcal{I}$, the family of continuous isotone functions.
%Here the condition `second-countable' can be dropped as it was used just in Lemma \ref{lem} to establish the existence of a continuous isotone function on $X^\times$ that  decreases over every fiber. This is provided directly by $-\pi_2$.

\begin{theorem} \label{mai2}
On a
locally compact $\sigma$-compact stably causal closed ordered space (causal structure) $(X,\mathscr{T},\le)$,  we have $\mathcal{I}\ne \emptyset$, $\mathscr{T}=\mathscr{T}_{\mathcal{I}}$, $\le=\le_{\mathcal{I}}$. Under second-countability we can replace the continuous isotone functions with the  time functions.
\end{theorem}

\begin{definition}
A spacetime is said to be $\mathcal{R}$-completely regular (with $\mathcal{R}$ the set of continuous rushing functions) if $\mathscr{T}=\mathscr{T}_{\mathcal{R}}$, $\le=\le_{\mathcal{R}}$ and $d=d_{\mathcal{R}}$.
\end{definition}

Thus we have proved the following converse of Theorem \ref{ckkr}
%second-countable

\begin{theorem}
The
%second-countable and
locally compact $\sigma$-compact stable spacetimes $(X,\mathscr{T},\le,d)$  are $\mathcal{R}$-completely regular.
\end{theorem}

In order to establish this converse we had just to impose local compactness and $\sigma$-compactness of the topology.

\begin{remark}
It is well known that on a Polish space $X$, the following  distance formula holds  
\[
c(x, y) =   \max_{\phi \in \textrm{Lip}_1(X)} \vert \phi(y) - \phi(x)  \vert
\]
where $c$ is a metric, continuous with respect to the Polish topology, and $\textrm{Lip}_1(X)$ is the family of 1-Lipschitz functions. This result can be proved via the Kanto\-ro\-vich-Ru\-binstein duality of Optimal Transport (OT) by selecting Dirac measures \cite{villani09}. Under suitable assumptions (e.g.\ existence of a continuous rushing function) an analogous distance formula can be proved in the Lorentzian case
\[
 \tau(x, y) =   \inf_{f \in \mathcal{B}} [f(y) - f(x)] ,
\]
where $\mathcal{B}$ is the set of rushing functions $f: X \to [-\infty, +\infty)$ not identically $-\infty$, such that $f(x)>-\infty$.
 The main difficulty here is that the representing functions though rushing are neither real (they can take the value $-\infty$) nor continuous because the property of {\em  proper $c$-concavity} is less demanding than in the metric case. As a result, the distance formula one gets is less precise than $d=d_{\mathcal{R}}$ proved above (and more related to that proved by Rennie and Whale \cite{rennie16}).
% and similar to that by Rennie and Whale \cite{rennie16}. 
Anyway OT  demands the Polish property while we assumed Hausdorffness, local compactness and $\sigma$-compactness. The Polish property does not imply local compactness or $\sigma$-compactness, and in the reverse direction one needs to assume second countability to infer the Polish property.
\end{remark}

\section{Conclusions}

In this work we introduced a new minimal notion of spacetime which encodes a topology, a causality structure and a Lorentzian distance. Within this setting it makes sense to speak of time functions (via the causal structure) and of proper time (via the Lorentzian distance). With respect to previous approaches we dropped continuity of $d$, working with upper semi-continuous regularity in analogy with our previous work on closed cone structures and closed Lorentz-Finsler spaces. Particularly interesting  is the category  of stable spacetimes, which in the smooth Lorentzian case correspond to the Lorentzian submanifolds of Minkowski spacetime \cite[Thm. 3.10, 4.13]{minguzzi17}\cite{minguzzi23b}.

The main result of this work is that any family of distinguishing functions over a set induces a stable spacetime and that, conversely, under local compactness and $\sigma$-compactness every stable spacetime arises in this way. The family of functions can be identified with the continuous rushing functions, and under second-countability, with the rushing time functions.

By the product trick, the continuous  rushing  functions are closely related to the translationally invariant continuous isotone functions for a space  $X^\times$ with one additional dimension. They have the advantage of encoding both metrical and causal properties of spacetime. We may say that the whole arena of Physics, namely the spacetime, emerges from the notion of time.

%\bibliography{../../bibliografie/simultaneity,../../bibliografie/libri,../../bibliografie/miei,../../bibliografie/mieiPrep,../../bibliografie/mieiProc}
%\bibliographystyle{plain}

\section*{\normalsize Data availability statement}
No new data were created or analysed in this study.

\section*{\normalsize Conflict of Interest statement}
The author of this publication declares no conflict of interest.

\end{document}